\documentclass[journal,cspaper,compsoc]{IEEEtran}
\usepackage{fixltx2e}
\usepackage[nocompress]{cite}
\usepackage[pdftex]{graphicx,graphics}
\usepackage[cmex10]{amsmath}
\usepackage{amssymb,amsthm}
\usepackage{url}
\usepackage{float}
\usepackage[inline]{enumitem}
\usepackage{subfig}	
\usepackage{array}
\usepackage{algorithmic}
\usepackage[linesnumbered, boxed, commentsnumbered, ruled, vlined]{algorithm2e}
\usepackage{makecell}
\usepackage{multirow}
\usepackage{booktabs}
\usepackage{color}
\usepackage{amsfonts}
\usepackage[normalem]{ulem}

\begin{document}

\newtheorem{definition}{Definition}[section]
\newtheorem{ex}{Example}[section]
\newtheorem{lemma}{Lemma}[section]
\newtheorem{theorem}{Theorem}[section]

\title{Risk-Aware Sensitive Property-Driven
Resource Management in Cloud Datacenters}

\author{
Muhamad Felemban, \textit{Member, IEEE}, Abdulrahman Almutairi, and Arif Ghafoor, \textit{Fellow, IEEE}
\IEEEcompsocitemizethanks{\IEEEcompsocthanksitem M. Felemban is with the Computer Engineering Department, the Information and Computer Science Department, and the Interdisciplinary Research Center for Intelligent Secure Systems at KFUPM. 
\IEEEcompsocthanksitem A. Ghafoor is with the School of Electrical and Computer Engineering and Purdue's Center for Education and Research in Information Assurance and Security~(CERIAS), Purdue University, West Lafayette,IN, \IEEEcompsocthanksitem A. Almutairi is with the Computer Engineering Department at the College of Computer and Information Science  in King Saud University, Riyadh, Saudi Arabia }}

\IEEEcompsoctitleabstractindextext{
\begin{abstract}
Organizations are increasingly moving towards the cloud computing paradigm, in which an on-demand access to a pool of shared configurable resources is provided. However, security challenges, which are particularly exacerbated by the multitenancy and virtualization features of cloud computing, present a major obstacle. In particular, sharing of resources among potentially untrusted tenants in access controlled cloud datacenters can result in increased risk of data leakage. To address such risk, we propose an efficient risk-aware sensitive property-driven virtual resource assignment mechanism for cloud datacenters. We have used two information-theoretic measures, i.e., KL-divergence and mutual information, to represent sensitive properties in the dataset. Based on the vulnerabilities of cloud architecture and the sensitive property profile, we have formulated the problem as a cost-drive optimization problem. The problem is shown to be NP-complete. Accordingly, we have proposed two heuristics and presented simulation based performance results for cloud datacenters with multiple sensitivity.
\end{abstract}
\begin{IEEEkeywords}
Cloud services, access control, risk assessment, vulnerability
\end{IEEEkeywords}}

\maketitle

\section{Introduction}\label{sec:introduction}
\IEEEPARstart{C}{loud} computing provides access to a pool of shared configurable resources such as datacenters, software, and infrastructure. Organizations utilize cloud computing to save the cost of acquiring and maintaining computing resources on-premise. The main feature of cloud computing is the ability to dynamically assign and release physical and virtual resources to customers based on their demand. However, security challenges present a major obstacle towards a complete shift to cloud computing paradigm \cite{DBLP:journals/cit/HuQLGTMBH11}. Cloud Security Alliance (CSA)\footnote{https://downloads.cloudsecurityalliance.org/assets/research/top-threats/Treacherous-12\_Cloud-Computing\_Top-Threats.pdf} identifies the top security threat in cloud computing as the breach of confidential data stored on the cloud. Such breach can be resulted due to human error, poor security practice, or software vulnerabilities. A recent example of such data breach is the leakage of US voter data, which was stored on a data warehouse service hosted by Amazon S3. The breach was caused by a misconfiguration of the database leaving the data exposed to the public \cite{upguard2017}. 

Data sharing among authorized users remains the main security concern among cloud tenants \cite{takabi2010security}. To mitigate the risk of data leakage, encryption-based techniques are used to preserve the confidentiality of the data stored on the cloud \cite{tang2016ensuring}. Recently, encryption primitive and homomorphic encryption techniques have been proposed to enable secure data search and computation on encrypted data. However, such techniques are computationally intensive and support limited operations. On the other hand, non-encryption based approaches to reduce the risk of data leakage by implementing resource isolation mechanisms, e.g., trusted virtual domain \cite{catuogno2010trusted}, secure hypervisor \cite{steinberg2010nova}, and Chinese wall policies \cite{berger2009security}. However, achieving resource isolation reduces resource utilization. 

In datacenter applications, the confidentiality of the shared data is often preserved using authorization mechanisms to manage data access \cite{alcaraz10}. In essence, cloud providers controls data operations using Role-Based Access Control (RBAC) policy, in which a user can perform an action on the shared data if the user is assigned a role and the role has the authorization to perform the action. For a particular datacenter application, cloud providers allocate a number of virtual resources, i.e., Virtual Machines (VMs), that are either hosted in a single or multiple Physical Machines (PMs) such that cost of resource provisioning is minimized. A user attempting to access the datacenter using an access control role is assigned to a particular VM subject to the Service Level Agreement (SLA). From security perspective, it is imperative that the assignment of virtual resources takes into account the risk associated with the data leakage resulted from assigning roles to VMs. In \cite{almutairi2014risk,almutairi2015risk}, risk-aware assignment techniques have been proposed with the objective to minimize the data leakage. 

In general, data contains embedded properties that can be extracted using data mining and information retrieval techniques. For example, social network dataset generally contains patterns that represent relationships among users \cite{cho2011friendship}, enterprise dataset contains association and business rules \cite{verykios2004association}, and healthcare data may contain infectious disease outbreak behavior \cite{lu2010prospective}. Some embedded properties can be sensitive and need to be preserved from authorized users. Furthermore, colocating roles in a single PM can increase the risk of disclosing the sensitive properties due to the risk of resource sharing among VMs. In this paper, we propose an efficient resource allocation technique for cloud datacenter applications with the aim to mitigate the risk associated with disclosing the sensitive properties. In particular, we assume that global sensitive properties are defined over the entire dataset, and each individual role perceives local sensitive properties associated with the role's permissible partial dataset. The local sensitive properties may disclose partial information about the global sensitive properties. To evaluate the information disclosure of the sensitive properties, we use information-theoretic measures to quantify the difference between global and local sensitive properties. 

\textbf{Motivating example:} Consider a cloud datacenter hosting a location-based social network application that allows to share information about the location and time of users' check-in activities. Furthermore, the access to check-in data is controlled using a fine-grained access control policy, e.g., Context-aware RBAC (CRBAC). In CRBAC, the time and location are used to define the context at which the access to check-in entries is controlled. Figure \ref{ch6:fig:map} depicts an example of a CRBAC policy with 5 roles with each role  is represented as a rectangle in the map. The rectangle boundary of the roles defines the spatial extent of the check-in entries during the last 24 hours. Several sensitive properties associated with the check-in data can be envisioned based on the demographics \cite{dong2014inferring} or commute distance \cite{cho2011friendship}. For example, the rate at which users visit a certain type of locations, e.g., restaurants, theaters, etc, during certain time of the day, e.g., morning, afternoon, or evening, is a sensitive property. Such a property can be modeled as a statistical function, e.g., probability mass function (pmf). For each role in Figure \ref{ch6:fig:map}, a local pmf of the sensitive property is computed and used to estimate the global pmf defined over the entire dataset. Subsequently, a role can exploit the vulnerabilities in the cloud architecture and attack the dataset assigned to other roles to gain more information about the global pmf. Any gain about the global pmf is a risk that needs to be controlled. 
\begin{figure}[t!]
\centering
\includegraphics[width = 0.5\textwidth]{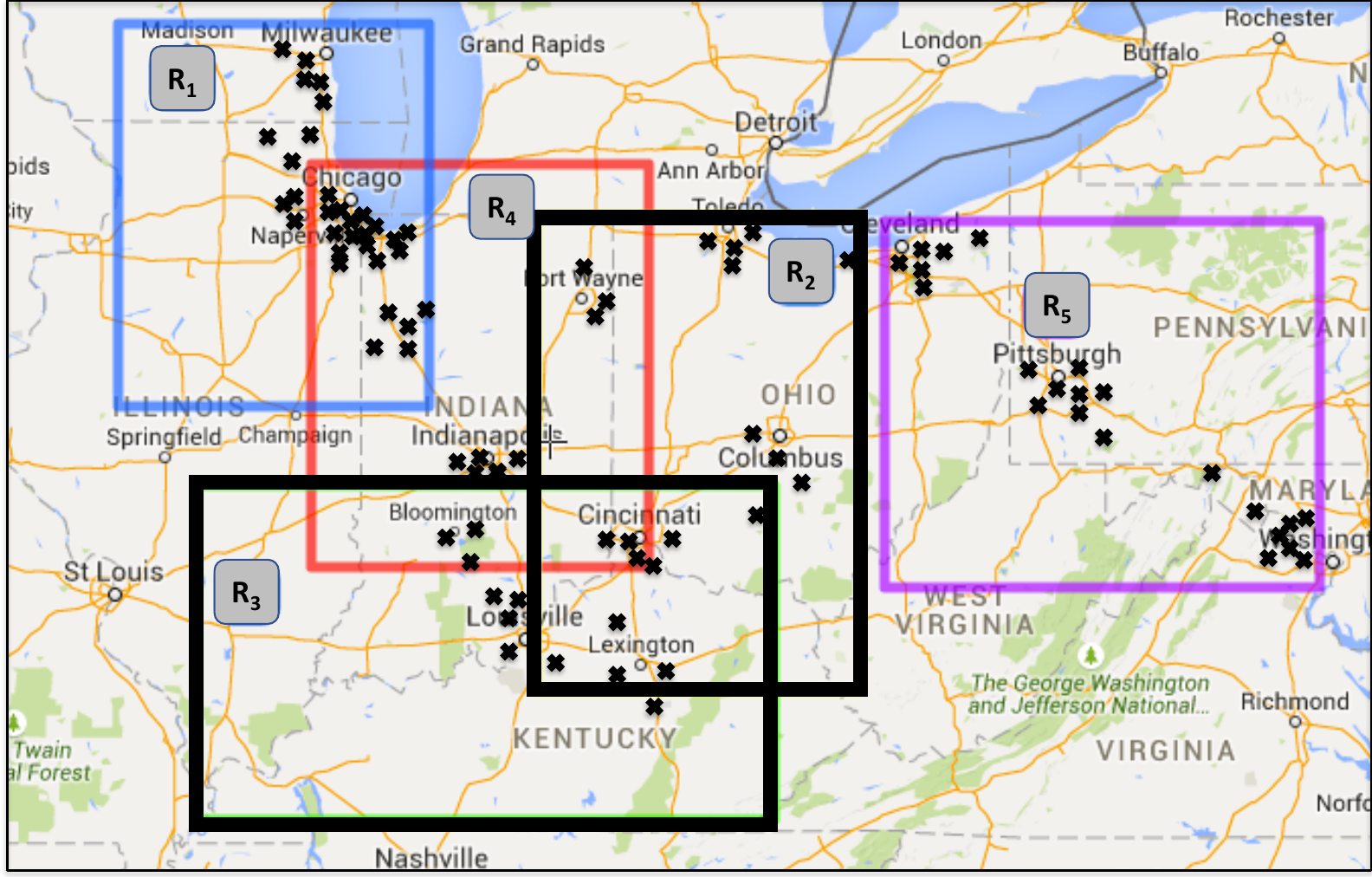}
\caption{Example of check-in data controlled by a CRBAC policy}
\label{ch6:fig:map}
\end{figure}

The contributions of this paper are as follows. First, we introduce a virtual resource management architecture that facilitates the assignment of virtual resources with objective to reduce the risk of data leakage. Accordingly, we formulate the risk-aware assignment problem that protects the global sensitive properties of the dataset. The problem is shown to be NP-complete and subsequently two heuristics are proposed. Finally, experimental evaluations of the proposed heuristics is conducted using real check-in dataset.

The remainder of this paper proceeds as follows. In Section 2, the vulnerability and authentication policy models are presented. Section 3 presents the risk-aware sensitive property-driven assignment problem. Section 4 presents the assignment heuristics. In section 5, we provide performance evaluation. Section 6 discusses related work. Finally, Section 10 outlines the conclusion.

\section{Cloud Vulnerability and Authorization Policy for Datacenters}\label{sec:profile}
In this section, we describe the risk associated with the data leakage that is caused by the assignment of access control roles to virtual resources in cloud datacenters. A statistical model of the access control policy with sensitive properties for cloud datacenters is introduced to control the assignment with the objective to reduce the risk of data leakage.

\begin{figure}[t!]
    \begin{center}
            \includegraphics[ width=0.5\textwidth]{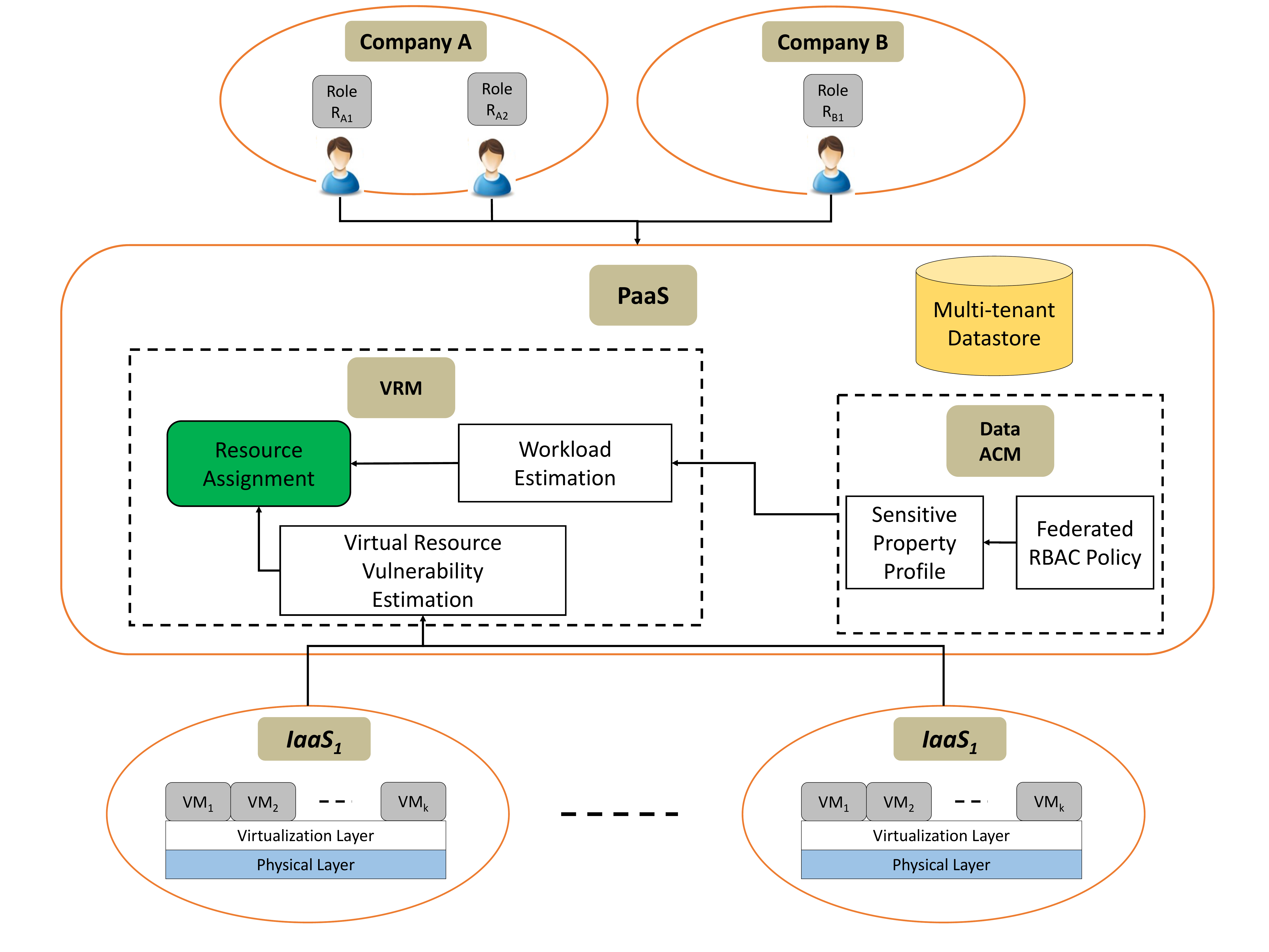}
        \caption{Virtual resource management architecture.}
        \vspace{-2.0em}
        \label{fig:arch}
       \end{center}         
\end{figure}

\begin{figure*}[t!]
    \begin{center}
        \subfloat[Example of RBAC permission assignment.]{\label{fig:rbac}
        \includegraphics[trim = 2mm 50mm 2mm 50mm, clip, width=0.4\textwidth]{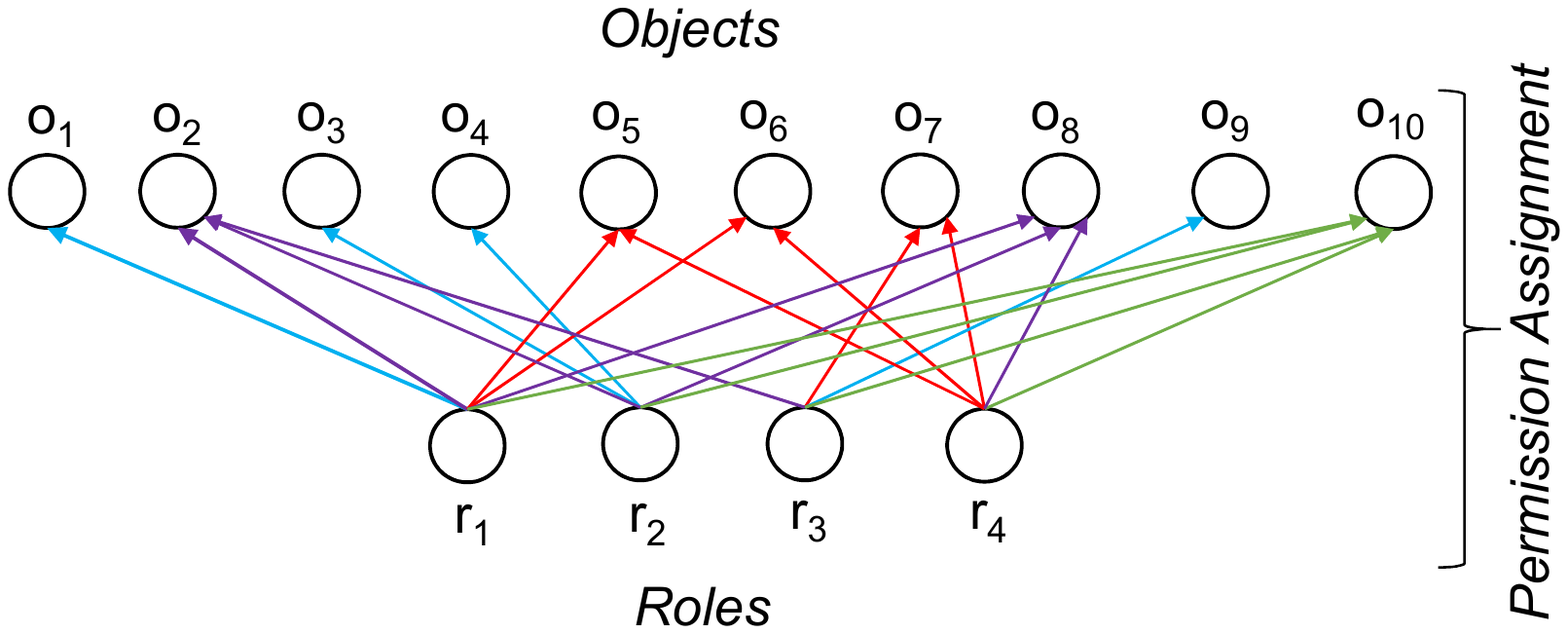}
        }
        \hskip 0.4truein
        \subfloat[Sensitive property profile of RBAC.]{\label{fig:lattice}
            \includegraphics[ clip, width=0.4\textwidth]{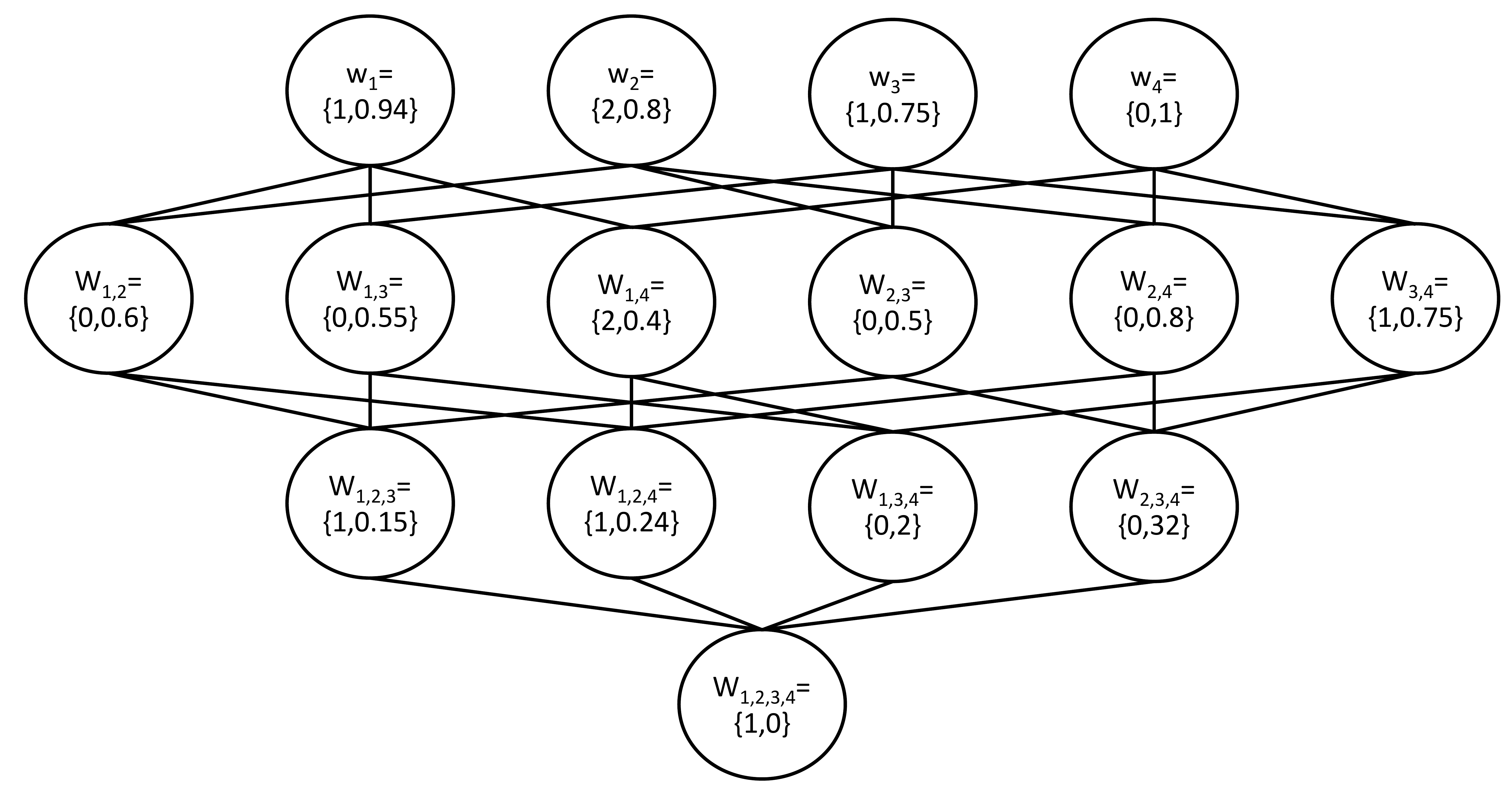}
        }
        \caption{RBAC policy representation.}
        \vspace{-2.0em}
       \end{center}         
\end{figure*}

\subsection{Cloud Vulnerability and Data Leakage Risk Model}
We take the perspective of Platform as a Service (PaaS) that allows developers to deploy Big Data applications. Examples of such applications can be found in healthcare, e-government, and enterprises \cite{kim2014big}. Virtualization technology is used to ensure isolation of applications while allowing efficient utilization of the physical resources provided by a single or multiple Infrastructure as a Service (IaaS) providers, i.e., Amazon Elastic Cloud Computing (EC2). Shared physical resources are provisioned to VMs such that the Quality Of Service (QoS) guarantees are satisfied. The QoS guarantees provided to the cloud customers is often documented in the Service Level Agreement (SLA). VMs are used to wrap the provisioned applications PaaS customers. 

In multitenant datacenters, the confidentiality of shared data is preserved using access control mechanisms, e.g., RBAC policy, that control the operations on data \cite{alcaraz10}. Under RBAC policy, a user can perform an action on the shared data if the user is assigned a role and the role has the authorization to perform the action. In \cite{almutairi2012distributed}, a distributed access control architecture that assigns virtual resources to cloud customers is proposed. The assignment is performed such that the cost of provisioning for PaaS cloud providers is minimized while satisfying the SLA for each cloud customer. Figure \ref{fig:arch} provides a system view of the cloud architecture for the virtual resource management. It consists of the datastore, Virtual Resource Manager (VRM), and Access Control Module (ACM). VRM consists of workload estimation, virtual resource vulnerability estimation, and resource assignment components. In this paper, we adopt the proposed architecture in  \cite{almutairi2012distributed} in order to develop a risk-aware sensitive property-driven virtual resource assignment framework.

Information security risk is defined by ISO 27005 as the following: \textit{"The potential that a given threat will exploit vulnerability of an asset or group of assets and thereby cause harm to the organization"}. Accordingly, the risk is qualitatively expressed as the product of the likelihood of the attack and its impact \cite{djemame2016risk}. The likelihood of an attack depends on the vulnerabilities of the system and the attacker threat, while the impact of data leakage depends on the quality or quantity of the leaked data. Formally, the data leakage risk can be formulated as the following:
\begin{equation}
Risk = Vulnerability \times Threat \times Assets 
\label{eq:risk}
\end{equation}
In particular, the vulnerability is the probability of data leakage between virtual resources. There are many software vulnerabilities that can cause data leakage across tenants, e.g., VM escape, SQL injection, and cross-site scripting \cite{pearce2013virtualization}. Such vulnerabilities have different levels of criticality which are calculated based on the Common Vulnerability Scoring System (CVSS) scores \cite{grobauer2011understanding}. The probabilities of finding each of these vulnerabilities are estimated using statistical methods such as Vulnerability Discovery Models (VDMs) \cite{ristenpart2009hey}. Consequently, the vulnerability probability is weighted based on its criticality and used to find the overall probability of data leakage. We model the vulnerabilities among virtual resources as a vulnerability matrix $\mathcal{D}$, where $d_{ij}$ is the data leakage vulnerability between the $i$-th and $j$-th resources. To capture the worst-case scenario for risk assessment, we assume that the threat is equal 1 for all roles. In other words, each cloud customer poses the same threat in accessing other customers' data objects, and vice versa. The assets are the data stored in the cloud datastore. The value of the assets in Eq. \ref{eq:risk} can be the number of leaked data items \cite{almutairi2015risk}. In this paper, however, the value of the assets is quantified as the amount of information gained about the sensitive property as a result of the attack. We use information-theoretic measure to quantify such information gain. More discussion is provided in Section \ref{sec:SP}.

\subsection{Context-aware Role-Based Access Control}
RBAC policy defines permissions on objects based on employee roles in an organization. The policy configuration is composed of a set of Users ($U$), a set of Roles ($R$), a set of Permissions ($P$), a user-to-role assignment relation ($UA$), and a role-to-permission assignment relation ($PA$). In a cloud datacenter, the RBAC policy defines the permissions to access data objects for roles assumed by the cloud users \cite{ferraiolo2001proposed}. A leading practical example of using RBAC in cloud datacenters is Microsoft Azure, which uses RBAC to authorize access to its resources through the Azure Active Directory. In addition, Azure controls its data operations using RBAC for performing create/read/update/delete operations in SQL DB \cite{Azure}. Another example is Oracle Sales Cloud that uses RBAC to secure the access to its data and functionalities \cite{oracle}. RBAC policy is modeled as a bi-partite graph as the following:
\begin{definition}
 Given an RBAC policy $ \mathcal{P}$ for a datacenter where $R$ is the set of roles and given $O$ being the set of data objects,  the role-to-permission assignment $PA$ can be represented as a directed bipartite graph $G(V,E)$, where $V=R\cup O$ s.t. $R\cap O=\phi$. The edges $e_{r_i o_j}\in E$ in $G$ represents the existence of role-to-permission assignment $(r_i\times o_j)\in PA$  in the RBAC policy $\mathcal{P}$, where $r_i \in R$ and $o_j \in O$.
\end{definition}

The out-degree of a role vertex  represents the cardinality of the role and the in-degree of a data object vertex represents the degree of sharing of that object among roles. In Figure \ref{fig:rbac}, an RBAC policy with $|R|=4$ and $|O|=10$ is represented as bipartite graph model. We can notice that the cardinality of role $r_1$ is $out$-$degree(r_1) = 6$. Also, the degree of sharing of data object $o_{10}$ is $in$-$degree(o_{10}) = 4$.

CRBAC is proposed to capture context parameters, such as location and time, in the access control policy \cite{samuel2008context}. Accordingly, the assignment of permissions to users are based on the context in order to provide a comprehensive approach for security management. For example, CRBAC is used for spatial and temporal access control for mobile devices \cite{shebaro2015context}. Spatial extensions have also been proposed to RBAC. One such extension is GEO-RBAC that defines spatial roles. A spatial role can be assumed within a specified spatial boundary. However, GEO-RBAC specifies one spatial extent for each role, which implies that each spatial location in an organization needs to have its own role. The GST-RBAC \cite{Arjmand15} is proposed to express the spatial information as spatial constraints which can be attached to any role already existed in an access control policy.

\subsection{Sensitive Property Profile (SPP) of RBAC}

In previous work, an alternative representation for RBAC, called Spectral Model (SM), is proposed \cite{almutairi2015risk}. In SM, data objects that are accessed by different roles is grouped into a set of non-overlapping partitions. Unlike the bipartite model, SM can be used to characterize the sensitivity of a datacenter using a single parameter, i.e., the degree of sharing among roles. An elaborated discussion about datacenter sensitivity will be provided in Section \ref{subsec:sens}. To model the information disclosure of each set of roles, we modify SM by augmenting the information disclosure of the sensitive property perceived by each set of roles and create the Sensitive Property Profile (SPP). Formally, the SPP is defined as follows:

\begin{definition}
Given a set of RBAC policy roles $R$ and its bipartite graph representation $G(V,E)$, let $\mathcal{P}(R)$ be the power set of $R$ excluding the null set  $\phi$. The sensitive property profile of the RBAC policy is the vector  $\mathcal{W}$, indexed by $\mathcal{P}(R)$ and lexicographically ordered. Formally, let $p \in \mathcal{P}(R)$ be a set of roles, then $w_p \in \mathcal{W}$ is defined as follows:
\begin{equation}
\begin{split}
&w_p= \{ C(w_p), c_1,\dots, c_{k} \} \nonumber
\end{split}
\end{equation}  
Where $C(w_p) = |\{o_k : o_k \in O | \forall r_i \in p \quad \exists  e_{r_i o_k} \in E \}$ is the cardinality of the set $w_p$ with $C(\mathcal{W}) = 2^n -1$, and $\{ c_1,\dots,c_k\}$ is a set of characteristics of $w_p$. There are no restrictions on the definition of the roles characteristics of the set of roles. However, it uniquely identifies the set $w_p$ from other sets.
\end{definition}  

In the following section, we provide detailed discussion about estimating the cardinality of each set $w_p$ in SPP. Then, we provide information-theoretic representations of the characteristics.

\subsection{Statistical Model for Cardinality Estimation of Roles}
\begin{algorithm}[t!]
\small
\SetNlSty{normal}{}{.}
\KwIn{Number of data objects ${|O|}$, number of roles $n$, constant$s$.}
\KwOut{spectral representation of RBAC $\mathcal{W}$.}
    Let $\mathcal{B} = \{B_1, \ldots, B_n\}$\  bucket array\; 
    \ForEach {$i = 1, \ldots, {|O|}$}{
       $\alpha$ = zipf($n$,$s$)\;
       $B_{\alpha}=B_{\alpha}+1$\;
    }
    \ForEach {$i = 1, \ldots, n$}{
     	\ForEach {$j = 1,\ldots, B_i$}{
		$\alpha$ = zipf(${n \choose i}$,$s$)\;
		map $\alpha$ to random partition in level $i$ call it $\hat{p}$\;
		$w_{\hat{p}}$ = $w_{\hat{p}}$ +1\;
	}
        add $w_{\hat{p}}$  to $\mathcal{W}$
    }
    return $\mathcal{W}$
   
\caption{Workload generation algorithm}

\label{alg:workload}
\end{algorithm}

In a datacenter with RBAC policy, specifying the exact cardinality of $w_p$ is a computationally intensive problem. One practical approach is to use a cardinality estimation technique \cite{zhang2009psalm}. In this paper, however, we assume that the data objects in a datacenter are accessed following a Zipfian distribution \cite{cooper2010benchmarking}. According to this distribution, the number of objects shared by roles exemplifies a power-law behavior, in which few objects are shared by a large number of roles while most of the objects are shared among a smaller number of roles; hence it can provide a heterogeneous workload for RBAC. The Zipfian distribution is given as follows:
\begin{equation}
\label{eq:zipf}
\begin{split}
&z(\alpha;s,N) = \frac{\alpha^{-s}}{\sum_{i=1}^N i^{-s}} \\
\end{split}
\end{equation}
where $N$ is the number of objects, $\alpha$ is the object rank, and $s \geq 1$ is the sensitivity parameter. If the sensitivity parameter $s=1$, then the probability that data object is assigned to a single role, i.e., with rank  $\alpha= 1$, is double the probability of that data object is assigned to two roles, i.e., with rank $\alpha =2$. As the value of $s$ increases, the number of data objects assigned to single roles becomes larger. 

The Zipfian distribution is used to generate heterogeneous RBAC-based workload in two steps as proposed in Algorithm \ref{alg:workload}.  In the first step, data objects are classified into $n$ buckets where each bucket represents the number of total data objects assigned to a lattice level. For example, data objects in bucket 1 are exclusively  accessed by only one role.  On the other hand, data objects in bucket $n$ are shared by all roles.  The number of data objects in each bucket follow Zipfian distribution. In the second step, we assign data objects of bucket  $i$ to randomly selected partitions at level $i$ of the lattice. Note, the number of partitions at level $i$ is $n \choose i$.

\subsubsection{Modeling Sensitivity of Cloud Datacenter}
\label{subsec:sens}
Based on the statistical property of the cardinality of roles in RBAC policy, we propose the notion of \textit{sensitivity} to classify cloud datacenters. The classification of cloud datacenters is dependent on the level of sharing of data objects among the roles. In particular, we define the sensitivity of a datacenter as the average degree of sharing among its data objects. In case the degree of sharing on average is low, we term the datacenter to have a \textit{high sensitivity}. On the contrary, if extensive sharing of data objects among roles is present, we term this datacenter to have a \textit{low sensitivity}. The \textit{medium sensitivity} class falls in the  middle. It can be noticed that the sharing of data objects and the classification of a datacenter can be modeled using the scalar parameter $s$ of the Zipfian density function shown in Equation \ref{eq:zipf}. As shown in Figure \ref{fig:lattice_graph}, for a smaller value of $s$, more data objects are uniformly distributed in the $\mathcal{W}$ vector of the spectral model of RBAC. In the following example, we illustrate how Zipfian parameter $s$ can be used to classify the sensitivity of datacenter. 

\begin{figure}[t!]
  \centering 
	\includegraphics[trim = 0mm 0mm 0mm 0mm, clip, width=0.5\textwidth]{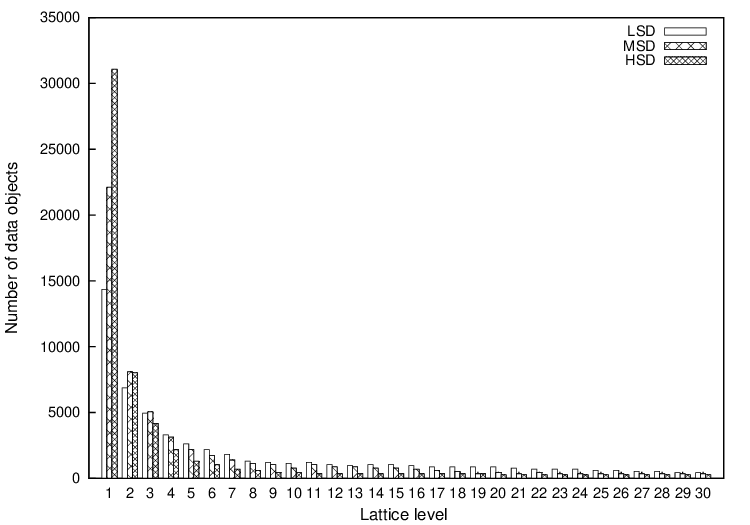}
  \caption{A statistical characterization of sensitivity of cloud datacenters}
  \vspace{-1.0em}
  \label{fig:lattice_graph}
\end{figure}


\begin{ex}
\label{ex:sensitivity}
For datacenter with $0.5 \times 10^6$ data objects, suppose we have three RBAC policies $(P_1,P_2,P_3)$ each with $n=30$ roles. Figure \ref{fig:lattice_graph} shows a histogram of object across the spectral lattice. Depending on the Zipfian distribution, the three classes of datacenters, namely; High Sensitive Datacenter (HSD), Medium Sensitive Datacenter (MSD), and Low Sensitive Datacenter (LSD), can be identified with respect to policies $P_1$, $P_2$, and $P_3$. For example, HSD has a large value of $s$  ($s \ge 2$) since the sharing of data objects among $P_1$ roles is very small. On other hand, LSD has a small value of $s$  ($1.5 > s \ge 1$) depicting the case of extensive sharing of data objects among roles of policy $P_3$. The MSD has a value of $s$ which falls in the  middle  ($2 > s \ge1.5$ ). Note, the number of data objects at level 1 in HSD is double the number of data objects at level 1 in LSD. 
\end{ex}

\subsection{Information-Theoretic Modeling of Sensitive Properties}
\label{sec:SP}

In this section, we present information-theoretic measures to quantify the information gain of the data leakage. The measures are used to quantify the assets in Eq. \ref{eq:risk}. The VRM component of Figure \ref{fig:arch} uses SPP to assign the RBAC policy roles to virtual resources. In essence, a function $f(.)$ that quantifies the disclosure of the global sensitive property is used along with the cardinality to characterize the set of roles. The function $f: \mathcal{P}(R) \rightarrow \mathbb{R}$ is a quantitative measure of the relative difference between the global sensitive property and local sensitive property obtained by accessing data according to the role set $A \in \mathcal{P}(R)$. In this paper, we consider two information-theoretic measures to model sensitive properties; Kullback-Leibler divergence ($KLD$) and Mutual Information ($MI$)  \cite{cover2012elements}. An important measure in information-theory is Entropy, which is the measure of \textit{uncertainty} in probability distributions. Let X be a random variable with probability distribution $P(X)$. The entropy $H(X)$ is defined as the following:
\begin{equation*}
H(X) = - \underset{x \in X}{\sum} p(x) \ \text{log} \ p(x)
\end{equation*}
 
\subsubsection{$KLD$-based Representation of Sensitive Property}

Given two discrete pmfs, $P(X)$ and $Q(X)$, $KLD$ is a non-symmetric measure of the difference between these two pms. Intuitively, $KLD$ measures the amount of information loss when $Q$ is used to approximate $P$. Formally, $KLD$ between two pmfs as the following:
\begin{equation*}
D(P||Q) = \sum_{x \in X} P(x) \  \text{log} \ \frac{P(x)}{Q(x)}
\end{equation*}


Using the definition of $KLD$, we define $KLD$-based property as the following:
\begin{equation}
f(A) = D(P_A || P_G)
\label{eq:f_div}
\end{equation}
Where $P_A$ is the joint probability distribution defined over data accessed by the set of roles $A$ and $P_G$ is the joint probability distribution over the entire dataset.

\begin{ex}
In Figure \ref{fig:lattice}, the general spectral model of the RBAC policy using $KLD$-based property is shown. It can be noticed that $w_{\{1,4\}} =\{ 2,0.4\}$, i.e., the cardinality of $w_{\{1,4\}}$ is 2 because $o_5$ and $o_{6}$ are accessed by both $r_1$ and $r_4$, while $f(w_{\{1,4\}})=0.4$ indicates that the probability distribution created using tuples $o_5$ and $o_6$ has a 0.4 KL divergence from the probability distribution created using $O$.  
\end{ex}


\subsubsection{$MI$-based Representation of Sensitive Property}
$MI$ is a measure that quantifies the dependency between two or more random variables. Intuitively, it measures the reduction of entropy in one random variable by observing the other. Formally, mutual information is the divergence between the joint pmf of two random variables and the product of their respective pmfs as defined below:
\begin{eqnarray*}\label{eq:mi}
MI(X;Y) &= &D(P(X,Y)||P(X)P(Y))\\
&=& \sum_{x \in X}\sum_{y \in Y} p(x,y) \text{log}\frac{p(x,y)}{p(x)p(y)}\\
& = & H(X) - H(X|Y)\\	
\end{eqnarray*}
It can be shown that the $MI$ approaches zero when the two random variables are independent, i.e., $H(X|Y) = H(X)$. On the other hand, strong dependency between $X$ and $Y$ results in a higher value of $MI$. 

Using the definition of $MI$, we define $MI$-based property as the following:
\begin{equation}
f(A) = |MI_A(X;Y) - MI_G(X;Y)|
\label{eq:f_mi}
\end{equation}
Where $MI_A(X,Y)$ is the mutual information between $X$ and $Y$ defined over data accessed by the set of roles $A$, and $MI_G(X,Y)$ is the mutual information between $X$ and $Y$ defined over the entire data set.

\section{Problem Definition and Formulation}\label{sec:assignment}
Given the proposed GSM associated with RBAC policy, and the virtual resource vulnerability model related to the cloud \cite{almutairi2014risk}, the Risk-aware Sensitive Property-driven Assignment Problem (RSPAP) is defined with the objective to minimize the risk (Eq. \ref{eq:risk}) of assigning virtual resources to access control roles as the following:

\begin{definition}\label{def:RSPAP}
(RSPAP). Given a suite of $m$ virtual resources, $v_1, \ldots, v_m$, an RBAC policy and the associated SPP , and the vulnerability matrix $\mathcal{D}$, the risk-aware sensitive property-driven assignment problem is to assign access control roles to virtual resources such that the total risk of sensitive property disclosure is minimized. Formally, let $Q_i = \{ A\in \mathcal{P}(R) | r_i \in A\}$ be the set of sets in $\mathcal{P}(R)$ that contains $r_i$ and $I$ be the assignment set of roles to VMs, where function $I(r_i)=v_j$ denoting the assignment of role $r_i$ to $v_j$. We formulate the $RSPAP$ optimization problem as the following:
\begin{equation}\label{eq:cost}
\begin{split}
\underset{I}{\text{Minimize}} \ & Risk \\ 
= &\sum_{r_i \in R} Risk(r_i) \\
= &\sum_{r_i \in R}  \max_{A \in Q_i} g_i^A \times \prod_{\underset{r_j \neq r_i}{r_j \in A}} d_{I(r_i),I(r_j)} 
\end{split}
\end{equation}
Where $g_i^A = |f(A)-f(r_i)|$ is the disclosure risk between role $r_i$ and set $A$, and $d_{k\ell}d_{I(r_i),I(r_j)} \in \mathcal{D}$ denotes the probability of leakage between $I(r_i)=v_k$ and  $I(r_j)=v_{\ell}$. 
\end{definition} 
\begin{theorem}
RSPAP problem is NP-complete.
\end{theorem} 

\begin{proof}
We construct a decision  formulation of RSPAP problem as follow:\\
For a given constant $c$, is there any assignment  $I$ such that 
\begin{eqnarray*}
Risk&=&\sum_{r_i \in R}  \max_{A \in Q_i} g_{i}^{A} \times \prod_{\underset{r_j \neq r_i}{r_j \in A}} d_{I(r_i),I(r_j)}   \le c
\end{eqnarray*}
First, given $c$ and $I$ we can check if the SPRAP decision problem satisfies the bound in polynomial time. Accordingly, the decision formulation of SPRAP is NP. Next, we show that decision SPRAP problem is NP-Complete by showing that  Travel Salesman problem (TSP) can be reduced to SPRAP. We formulates TSP for a given set $U$ with a distance metric $h(a,b)$ where $a,b \in U$, is there such an ordering of elements of $U$:($u_1,u_2,...,u_m$) such that $\sum_{i=1}^n h(u_i,u_{i+1})+h(u_n,u_1) \leq c$.

Let  $f$ be the sensitive property function for  RBAC policy with $n$ roles where $R= \lbrace r_1,r_2, \dots , r_n \rbrace$.  We define $f$ as follows:
\[ f(A) = \left  \{
  \begin{array}{l l}
    1 & \quad  \text{if} \ A = \lbrace r_i, r_j \rbrace\ \text{and} \ |i-j| \ mod \ n = 1\\
    0 & \quad \text{Otherwise}
\end{array} \right.\]

Let  $V = U$  and $d_{i,i} = h(a,b)$ where $v_i= a \ \text{and} \ v_j = b$ then assume the assignment of role $r_i$ to VM $v_j$ as if the salesman visits city $u_j$. Therefore, if we find an optimal order of visits in TSP, we can find the optimal assignment of $R$  to $V$ in RSPAP.
\end{proof} 

%
%

\section{Proposed Heuristics for RSPAP}\label{sec:heuristics}
In this section, we propose two heuristics for solving RKAP which can be used by the resource assignment component in VRM in Figure \ref{fig:arch}. The first heuristic follows a top-down clustering based approach. In each step, the cluster of roles is divided based on the risk associated with the disclosure of the sensitive property. The second heuristic is a neighbor-based heuristic, which uses a pairwise property disclosure measure for role scheduling. This measure is computed based on the inference function $f$. Each role is assigned to the best available virtual resource with respect to the probability of leakage.

\subsection{Top-Down Heuristic (TDH)}
In  this heuristic, a top down clustering approach is used. Initially, all the roles are assumed to be in one cluster. We begin with division of this root cluster and split it into two clusters such that total measure of property disclosure $g_i^A $ of both clusters is minimum. The resulting clusters are sorted based on their total $g_i^A$. TDH splits a cluster with the largest total $g_i^A$ further by moving roles with high value of $g_i^A$ from current cluster to the new cluster such that the total $g_i^A$ of both clusters is minimum. TDH repeats the cluster  splitting  step until it generates up to $m$ clusters that equals  the number of VMs. Subsequently, the cluster with the largest total $g_i^A$ is assigned to the VM with the minimum probability of leakage. After this initial assignment, TDH iterates over all the roles and changes the assignment of a role if it results in reduction of the total risk. The algorithm stops when any futher change in the assignment does not reduce the total  risk. 

 Algorithm \ref{ch6:alg:tdh}, formally represents the TDH. In Lines 1-2,  the initial cluster list has one cluster which is the set of all roles. In Lines 3-13,  TDH iterates the outer loops $m$ times and in each iteration the first cluster in  $\mathcal{C}$ is divided  into two new clusters $C1$ and $C2$. In the inner loop Lines 7-11, TBH moves the role $r_i$ from the current cluster $C1$ to the new cluster $C2$. The condition in Line 8 guarantees that this move does not increase the total  risk.  In Lines 12-13, TDH adds the new clusters to $\mathcal{C}$ and sorts the clusters based on  their  total $g_i^A$. Then,  the VMs  are sorted based on the intra probability of leakage in Line 14. The initial assignment is generated in Lines 15-17 by iterating over all the sorted VMs and the resulting sorted clusters.  TDH in Lines 18-22 iterates over all roles to improve the assignment. The final the assignment is stored in assignment matrix $I$.

\begin{lemma}
The  complexity of TDH is $O(n^3 \times m^2 )$
\end{lemma}
\begin{proof}
In TDH we need to generate $m$ clusters. For each of the cluster, the algorithm iterates over at most $n$ roles. For each of the roles the algorithm evaluate function  $f$ to compute the total disclosure risk of a cluster as illustrated in Line 8 in Algorithm \ref{ch6:alg:tdh}. The computation of $f$ has the complexity of $n^2$. Therefore, the complexity of TDH for computing the initial assignment is $O(n^3 \times m )$. To improve the assignment the algorithm iterates $n \times m$ times until no further improvement can be achieved. Each iteration requires $n^2\times m$ steps to calculate the total risk. Accordingly the complexity of the improvement step  is $O(n^3 \times m^2 )$ which results in the overall  complexity of TDH which is $O(n^3 \times m^2 )$.
  \end{proof}

\begin{algorithm}[t!]

\small
\SetNlSty{normal}{}{.}
\KwIn{sensitive property function $f$, vulnerability matrix $\mathcal{D}$.}
\KwOut{An assignment $I$ of roles to VMs.}
    Let $C1=\{R\}$ be the initial cluster\;
    Let $\mathcal{C} = \{C1\}$ be the set of clusters\;
   \While{$|\mathcal{C}| \neq m$}{
   	Let $C1 = \mathcal{C}[0]$ be the cluster with maximum property disclosure\;
	Let $C2 = \{\}$\;
	Let $dis = f(C1)$\;
		\ForEach{$r_i \in C1$}{
			\If{$f(C1-r_i) + f(C2 \cup r_i) < dis$}{
				$C1 = C1 - r_i$ \;
				$C2 = C2 \cap r_i$ \;
				$dis = f(C1) + f(C2)$\;
			}	
		}
	$\mathcal{C} = \mathcal{C} \cup C2$\;
	Sort $\mathcal{C}$ based on disclosure risk\;	
  }
 Let $B$ be the sorted list of VMs based on $d_{q,q}$\;	
 \ForEach{$i =1, \dot,m $}{
  	\ForEach{$r \in \mathcal{C}[i]$}{
		$I(r) = B[i]$\;  	
	}
} 
\ForEach{$r_i \in R$}{
	Let $t$ be the current risk for $r_i$\;
	\ForEach{$v_q \in V$}{
		\If{risk of $r_i$ when $I(r_i) = v_q$ less than $t$}{
			 $I(r_i) = v_q$\;
		} 
	}
} 
 return $I$\;

\caption{TDH}
\label{ch6:alg:tdh}
\end{algorithm}		

\subsection{Neighbor-Based Heuristic (NBH)}
The NBH algorithm is shown as Algorithm \ref{ch6:alg:nbh}. Prior to the assignment, NBH computes disclosure weights for each pair of roles as shown on Line  2-4 of the Algorithm \ref{ch6:alg:nbh}. These weights are based on the sensitive property function $f$ discussed in Section \ref{sec:SP}.  Subsequently, NBH follows a best fit strategy  whereby it initially selects the pair of roles  $(r_i,r_j)$ with the maximum weight and assigns them to the pair of VMs $(v_l,v_q)$  that has the least probability of leakage. It then selects the role $r_k$ that has the highest disclosure weight with any of perviously assigned roles.  $r_k$ is then assigned to  the VM that has the minimum probability of leakage to an already selected VM.  This step is repeated  until each VM is assigned exactly one role. The remaining $(n-m)$  roles are assigned to VMs based on the value of intra probability of leakage. Th reason for only considering intra probability is that we assume  that the probability of leakage  across VMs is always less than probability of leakage within a VM. 

The assigned roles are stored in list $A$. The unassigned  roles are kept in the list $F$. The list $G$ maintains IDs of the unassigned VMs. In Lines 12-18 , each iteration of  the loop assigns a role from list $F$ with the maximum value in matrix $C$ to a VM in $G$. Then the algorithm updates the assignment matrix $I$ and the lists $A$, $F$, and $G$. The outer loop in Lines 19-23 iterates over all the remaining unassigned roles to find the current best assignment. To find the best assignment for a given role in $F$, the inner loop (Lines 20-21) iterates over all VMs  and identifies the  VM that results in the minimum risk for that role. At the end NBH returns the assignment matrix.

\begin{lemma}
The  complexity of NBH is $O(n^2 \times m)$
\end{lemma}
\begin{proof}
The complexity of computing  $B_{i,q} $ in Line 22 of the NBH is $n \times m$. This  computational cost needs to be repeated $n$ times for each role in $F$. Therefore, the total complexity of NBH is  $O(n^2 \times m)$.

\end{proof}

\begin{algorithm}[t!]
\small
\SetNlSty{normal}{}{.}
\KwIn{sensitive property function $f$, vulnerability matrix $\mathcal{D}$.}
\KwOut{An assignment  $I$ of roles to VMs.}
   Let $C_{i,j}$ be the increase  in the measure of property disclosure due to the pair $(r_i,r_j)$ \; 
    \ForEach {$r_i \in R$}{
	\ForEach {$r_j \in R$}{
		$C_{i,j} =|f(r_i,r_j) - f(r_i)| +|f(r_i,r_j) - f(r_j)|$\;
		}
	}
    Find initial pair of vms   $(v_q,v_l)$ with minimum $d_{q,l}$\;
    Find initial pair of roles  $(r_i,r_j)$  with largest  $C_{i,j}$\;
    $I(r_i)=v_q$ and $I(r_j)=v_l$\;
     Let $A =\{ r_i,r_j\}$ be the set of assigned roles\;
     Let $F= R-\{r_i,r_j\}$ be the set of free roles\;
     Let $G= V-\{v_q,v_l\}$ be the set of free VMs\; 
     \While {G $\neq \phi$}{
     	Let $(r_i,r_j)$ be maximum $C_{i,j}$ where $r_i \in A$ and $r_j \in F$\;
	Find $(v_q,v_l)$ with minimum $d_{i,j}$ where $v_l\in F$  and  $I(r_i)=v_q$\;
	$I(r_j)=v_l$\;
	$A =A\cup \{ r_j\}$\;
	 $F =F -\{ r_j\}$\;
	 $G =G -\{ v_l\}$\;
	}
	 \ForEach {$r_i \in F$}{
	 	 \ForEach {$v_q \in V$}{
		 	\ForEach{$r_j \in A$}{
		 		Let $B_{i,q} =  C_{i,j} \times d_{q,q} $ such that $I(r_j)=v_q$\;
			}
		}
		 Let $r_i,v_q$ be the minimum  $B_{i,q}$\;
		$I(r_i)=v_q$\;
		$A =A\cup \{ r_i\}$\;
		 $F =F -\{ r_i\}$\;
	 }
  return $I$\;

\caption{NBH}
\label{ch6:alg:nbh}
\end{algorithm}

\section{Performance Evaluation}\label{sec:evaluation}
In this section, we present a detailed performance evaluation of the proposed heuristics with various values of $s$ (the sensitivity parameter of the Zipfian distribution), the number of virtual resources $m$, and the number of roles $n$. 

 \begin{table}[t!]
      \centering
       \begin{tabular}{| c | c | c | c |}
        \hline 
            x & Location type & x & Location type \\ \hline \hline
            1 & Accomodations & 9 &  Religious buildings\\ \hline
            2 & Governmental & 10 & City buildings  \\ \hline
            3 & Cemeteries & 11 & Shops  \\ \hline
            4 & Educational & 12 & Stadiums  \\ \hline
            5 & Restaurants & 13 & Recreational \\ \hline
            6 & Health centers & 14 & Transportation stops \\ \hline
            7 & Parks & 15 & Dams and towers \\ \hline
            8 & Banks  &&\\
            \hline
        \end{tabular}
        \caption{POI Location types}
        \label{tab:loc}
\end{table}
\begin{table}[t!]
      \centering
       \begin{tabular}{|c| c|}
        \hline
            y & Time slot \\ \hline \hline
            1 & 6:00 AM to 11:59 AM \\ \hline
            2 & 12:00 PM to 4:59 PM \\ \hline
            3 & 5:00 PM to 7:59 PM \\\hline
            4 & 8:00 PM to 5:59 AM\\\hline
        \end{tabular}
        \caption{Time slots}
        \label{tab:time}
\end{table}

\begin{figure*}[t]
    \begin{center}
        \subfloat[Global p.m.f]{\label{fig:global}
        \includegraphics[width=0.33\textwidth]{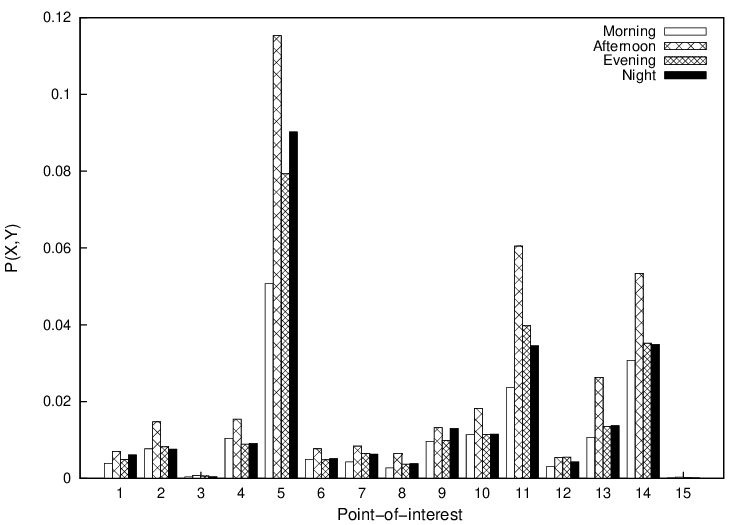}
        } 
     \subfloat[$D_1$ p.m.f]{\label{fig:role1}
        \includegraphics[width=0.33\textwidth]{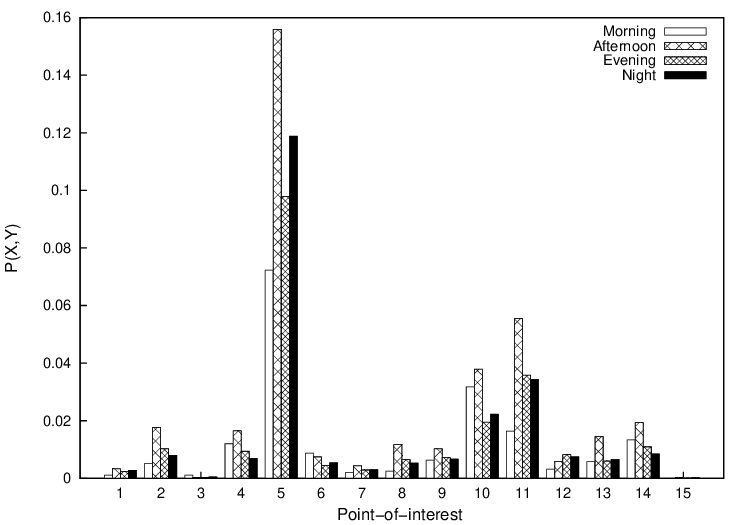}

        } 
            \subfloat[$D_2$ p.m.f]{\label{fig:role2}
        \includegraphics[width=0.33\textwidth]{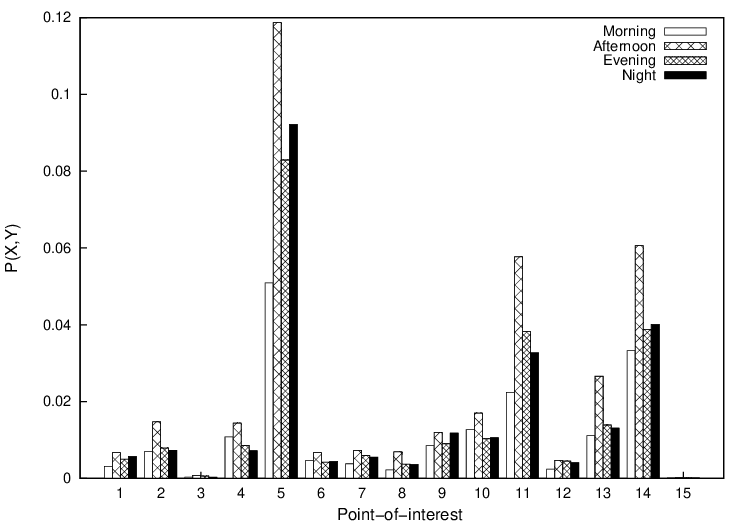}
        }
        \caption{Global and local p.m.f.}
\label{fig:pmf}        
       \end{center}         
\end{figure*}

\subsection{Check-in Dataset}
In this section we present a detailed discussion about the dataset and the sensitive properties used in evaluating the proposed heuristics. The experiments have been conducted using a  real life check-in dataset collected from \textbf{Gowalla} social networking website \footnote{https://snap.stanford.edu/data/loc-gowalla.html}. The Gowalla dataset has around 6.5 million check-in entries for about 200,000 users around the world over. The data was collected over a period of 20 months. The dataset contains information about time and geo-coordinates of the check-ins. In order to introduce location semantic in the dataset, we have extracted a list of Point-Of-Interests (POI) from \textbf{OpenStreetMap} \footnote{https://www.openstreetmap.org/} and map each check-in entry to the nearest POI based on geo-coordinates. The mapping is conducted as the following. The area of the check-in entry is increased to a circle of around 111 $m$ in radius by truncating the values of the latitude and longitude to 3 decimal digits. The candidate POIs are those which have the same values of the first 3 decimal digits. The POI with the smallest value in the 4th decimal digits of latitude and longitude is used. Each entry in the mapped dataset contains the time of check-in, the type of POI, and user ID. For the experiments in this paper, a subset of the mapped dataset that belongs to the US Midwest region is used. For this region of 12 US states, the overall size is about 250,000 entries.

\subsubsection{Definition of Sensitive Properties}

Two discrete random variables in the dataset is defined; $X$ and $Y$. Let $X$ be a discrete random variable that indicates the type of location of the check-in entry, and $Y$ a discrete random variable that indicates the time slot. The probability $Pr(X=x)$ (or $P_X(x)$) is the probability of check-in at a location of type x. The probability computed as the following:
\begin{equation*}
P_X(x) = \frac{\text{\# of check-ins in locations of type x}}{\text{Total number of check-ins}}
\end{equation*}

Similarly, $Pr(Y=y)$ (or $P_Y(y)$) is the probability of check-in during time slot y. The probability is computed as the following:
\begin{equation*}
P_Y(y) = \frac{\text{\# of check-ins during slot y}}{\text{Total number of check-ins}}
\end{equation*}
The joint probability $Pr(X=x,Y=y)$ (or $P_{X,Y}(x,y)$) is the probability of check-in at a locations of type x during time slot y. The joint probability is computed as the following: 
{\small
\begin{equation*}
P_{X,Y}(x,y) = \frac{\text{\# of check-ins in locations of type x during slot y}}{\text{Total number of check-ins}}
\end{equation*}
}
Note, that there are many ways to define the random variables $X$, $Y$, and their joint pmf. The above definition is used in order to study the correlation between the random variables representing the location and time of check-ins.

In our experiments, we categorize POIs into 15 categories and time into 4 slots. Tables \ref{tab:loc} and \ref{tab:time} summarize the discrete values of $X$ and $Y$. For the dataset associated with the Midwest region, Figure \ref{fig:global} depicts the joint pmf $p(x,y)$ for the random variables  $X$ and $Y$ taking values according to Table \ref{tab:loc} and Table \ref{tab:time}, respectively. To exemplify the effect of RBAC policy on the pmf, suppose there are two roles, $R_1$ and $R_2$, with dataset $D_1$ and $D_2$, respectively. The size of $D_1$ is 30400, while the size of $D_2$ is 138276. Note, $D_1$ and $D_2$ can be overlapping, i.e., can share some check-in entries. Figures \ref{fig:role1} and \ref{fig:role2} depict the pmfs of $D_1$ and $D_2$, respectively. For example, the probability of check-in in restaurants between 5 PM and 8 PM is 0.079, 0.098, and 0.083, for the entire dataset,  $D_1$, and  $D_2$, respectively.

\subsubsection{Monotonicity of the Sensitive Properties}

Figure \ref{fig:mono} show the relationship between the size and the embedded sensitive property of the check-in dataset. In this experiment, we vary the size of dataset between 10\% to 100\% with an increase of 10\% and choose check-in entries randomly. The sensitive property of the subset is computed. The experiment is repeated 10 times and the average values of divergence and mutual information are reported. We note that the divergence and mutual information are monotonically decreasing as the size of the dataset increasing. Note that this behavior might not hold for other dataset. 

\begin{figure}[t!]
    \begin{center}
        \subfloat[Divergence]{\label{fig:divMono}
        \includegraphics[width=0.4\textwidth]{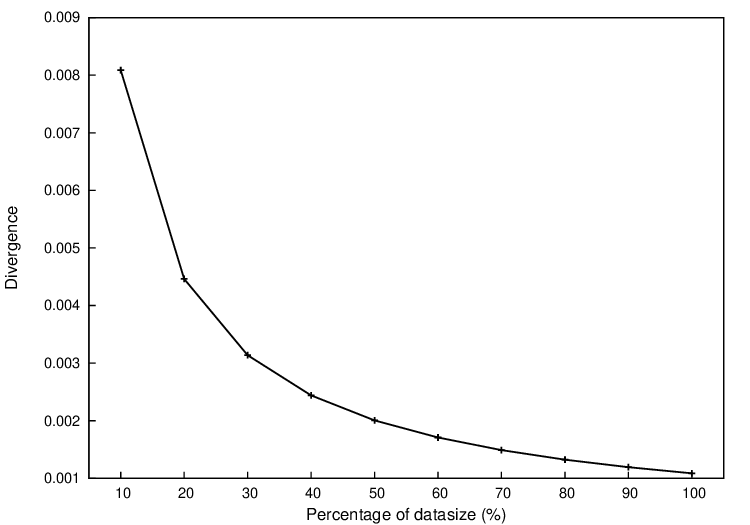}
        }
                \hskip 0.05truein
     \subfloat[Mutual information]{\label{fig:MImono}
        \includegraphics[width=0.4\textwidth]{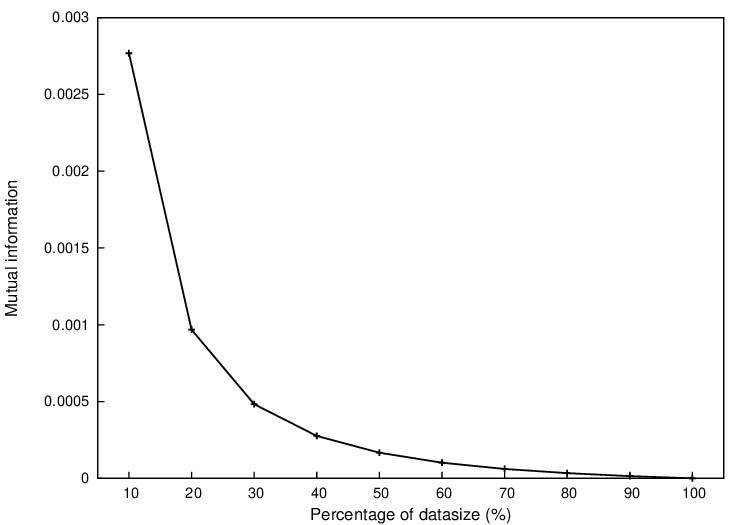}

        }
        \caption{Monotonicity of the divergence and mutual information in check-in dataset.}
\label{fig:mono}
       \end{center}         
\end{figure} 

\subsection{Evaluation Metrics}
In this section, we introduce three evaluation metrics to compare the performance of the proposed heuristics. The first two metrics evalaute with the overall disclosure risk exhibited by all roles. The third metric focuses on the disclosure risk incurred per role. The first metric is the \textit{disclosure risk}, i.e., $Risk$ Equation \ref{eq:cost}. This metric depends on the sensitive property function used in the cost function. The risk metric associated with the $KLD$-based sensitive property measures the risk incurred due to the reduction of the divergence between the attacker's pmf and the global pmf. On the other hand, the risk metric corresponding to the $MI$-based sensitive property measures the risk incurred due to the reduction in the difference between the estimated $MI$ by a role and the global $MI$. 

The second measure is the \emph{quality of risk-reduction} $(\Delta)$ that measures the ability of the proposes heuristics to reduce the disclosure risk by assigning roles to virtual resources. In order to compute $\Delta$, we define the property attackability ($PA$) as the maximum disclosure risk of the sensitive property for a given policy. Formally, $PA$ is given as the following: 
\begin{equation}
PA = \sum_{r_i \in R} f(r_i)
\end{equation}
 Accordingly, the quality of risk-reduction is the relative difference between $PA$ and the disclosure risk of the sensitive properties, i.e., $KLD$-based sensitive property or $MI$-based sensitive property. Formally, $\Delta$ is given as the following:
\begin{equation}
\label{eq:QoR}
\Delta = \frac{PA - Risk}{PA}
\end{equation}
Note, $ 1 \le \Delta \le 0$.

The third metric is the \textit{discriminator index} ($DI$)  which evaluates the performance of the heuristics at the role level. In particular, $DI$ is an indirect indicator of the performance in terms of risk reduction per role. Semantically, $DI$ not only indicates the quality of reducing the risk at the role level but also implies the disproportionality in terms of risk management among the roles. DI is an extension of the generic discriminator index \cite{jain1984quantitative}, and is given as the following:
\begin{equation}
\label{eq:DI}
 DI = 1- \frac{(\sum_{i=1}^n \Delta_i)^2}{n\times \sum_{i=1}^{n}(\Delta_i)^2} 
\end{equation}
where $\Delta_i$ is given as the following:
\begin{eqnarray*}
\Delta_i &= &\frac{f(r_i) - Risk(r_i)}{f(r_i)}
\end{eqnarray*}
and $Risk(r_i)$ is given in Equation \ref{eq:cost}.
%

\subsection{Experimental Setup}
In this experiment, we simulate up to 6 clusters of virtual resources, i.e., VMs. The size of a each cluster ranges from 1 to 32 VMs. A cluster of VMs can be hosted by one or more physical servers. However, in our experiment, we consider one physical server per virtual cluster. The inter-vulnerabilities across virtual clusters are zeros, thus reflecting  physical isolation. The intra-vulnerabilities among VMs within a cluster are randomly generated. All intra-VM vulnerabilities are non-zeros and thus higher than inter-VM vulnerabilities. This is correct because the fact that, in general, physical isolation is more secure than VM isolation. Two sensitive properties functions, i.e., $KLD$-based and $MI$-based, are used for the performance evaluation. Finally, we use the RBAC policy sensitive property profile truncated up to level 3.

\subsection{Performance Results for $KLD$-based Sensitive Property}
In this section, we discuss the performance of the proposed assignment heuristics for a $KLD$-based sensitive property with respect to the aforementioned performance metrics i.e. $Risk$, $\Delta$, and $DI$. 

\subsubsection{Total Risk ($RISK_{KLD}$)}
For both proposed heuristics, we observe that $RISK_{KLD}$ increases as we move from \textit{LSD} to \textit{HSD} case as shown in Figure \ref{fig:30vRiskDiv}. In other words, for the same value of divergence risk, the number of roles in $LSD$ required to obtain the risk value is more than the number of roles in $HSD$ as shown with the dotted horizontal lines in Figure \ref{fig:30vRiskDiv}. The reason can be explained using the monotonicity property shown in Figure \ref{fig:mono}. In particular, a role in $HSD$ workload gain more information by attacking other roles in compared to an $LSD$ workload. This is because that roles in $HSD$ workload share less data compared to $LSD$ workload and thus the data size increases. In general, we notice that $TDH$ outperforms $NBH$ since $TDH$ assigns role clusters based on aggregated risk  in contrast to $NBH$ which assigns each role independently. Figure \ref{fig:30vRiskDiv} shows that as we increase the number of roles, the resulting risk $RISK_{KLD}$ increases. This is expected as the total risk is the sum of  the risk from all the roles. 

However, $RISK_{KLD}$  decreases as the number of VMs used in the experiment increases as depicted in Figure \ref{fig:150rRiskDiv}. This behavior is expected as increasing the number of available resources, i.e. VMs, reduces the number of roles assigned to the same VM which reduces the intra VM attack. Also, we can notice for Figure \ref{fig:150rRiskDiv} that for the case of $HSD$ the decrease in $RISK_{KLD}$ is more  prominent than for the case of $LSD$. The reason being the parameter $PA$ has a higher value for $HSD$ due to low sharing of data among the roles. Therefore, increasing the number of VMs is more beneficial for the case of $HSD$ than for  the case of $LSD$.

 \begin{figure}[t!]
    \begin{center}
        \includegraphics[width=0.4\textwidth]{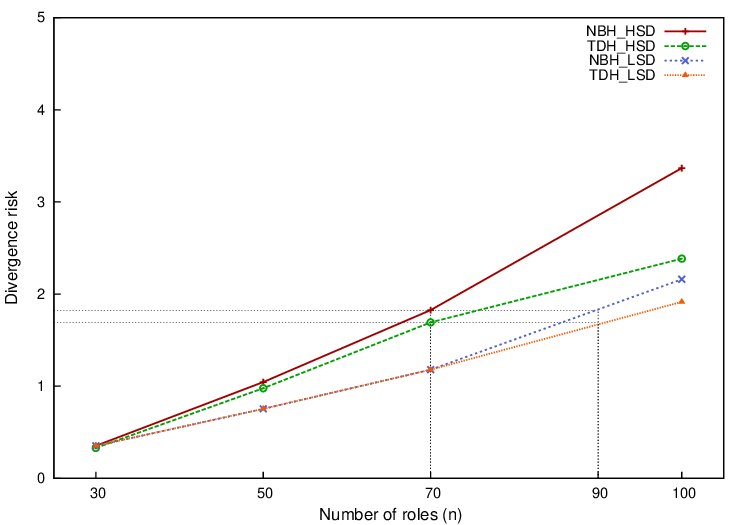}
        \caption{Total Divergence Risk $(RISK_{KLD})$ with a problem size of 30 VMs  for LSD and HSD datacenters.} 
        \label{fig:30vRiskDiv}     
             \end{center}          
                \vspace{-1.0em}
\end{figure} 

\begin{figure}[t!]
    \begin{center}
        \includegraphics[width=0.4\textwidth]{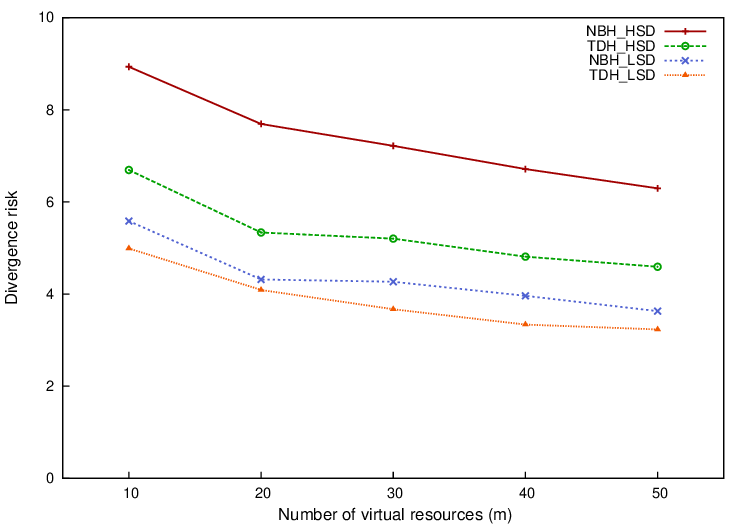}
        \caption{Total Divergence Risk $(RISK_{KLD})$ with a problem size of 150 roles for LSD and HSD datacenters.} 
        \label{fig:150rRiskDiv}
     \end{center}         
             \vspace{-1.0em}
\end{figure}


\subsubsection{Quality of Risk-Reduction($\Delta$)}
For a given heuristic, $\Delta$  decreases with $n$ as depicted in Figure \ref{fig:30vQoRDiv}. The reason being that by increasing the value of $n$, $PA$ of the RBAC policy increases relatively at a faster rate than the rate at which  $RISK_{KLD}$ increases. In addition, as shown in Figure \ref{fig:30vQoRDiv}, $TDH$  outperforms  $NBH$ for both LSD and HSD cases since  $TDH$ extensively searches for better assignments to improve  $RISK_{KLD}$. However in the case of HSD since the size of roles in term of  data items is smaller as compared to LSD, "partial sensitive property" possessed by a role is relatively small for $HSD$.  Therefore it is expected that the value of $\Delta$ is relatively low for HSD as compared to LSD; a phenomenon clearly observable from Figure \ref{fig:30vQoRDiv}.

In Figure \ref{fig:150rQoRDiv}, we observe that for given RBAC policy the increase in the number of VMs improves  $\Delta$ due to reduction in  $RISK_{KLD}$  while the value of parameter $PA$ for the policy does not change.   The behavior of the curves in terms of HSD and LSD in the Figure \ref{fig:150rQoRDiv} can be reasoned similar to reasons given for the Figure \ref{fig:30vQoRDiv}.

 \begin{figure}[t!]
    \begin{center}
        \includegraphics[width=0.4\textwidth]{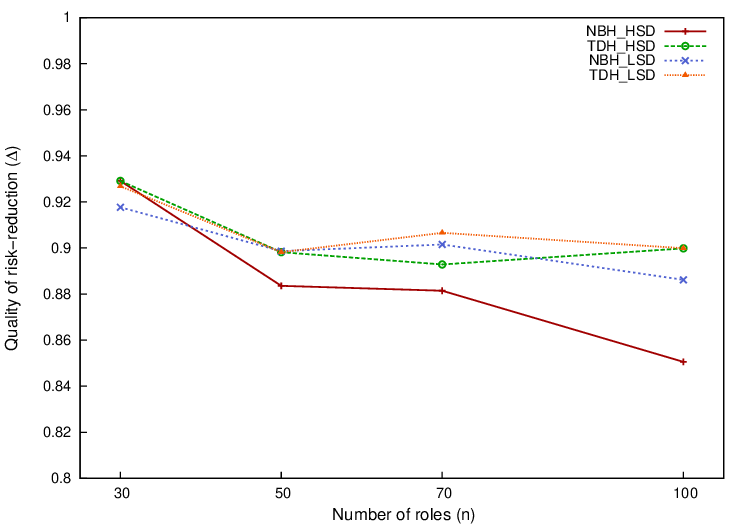}
        \caption{Quality of Risk-Reduction  $\Delta$ with a problem size of 30 VMs  for LSD, and HSD datacenters.} 
        \label{fig:30vQoRDiv}
        \vspace{-1.0em}
     \end{center}         
\end{figure} 

\begin{figure}[t!]
    \begin{center}
        \includegraphics[width=0.4\textwidth]{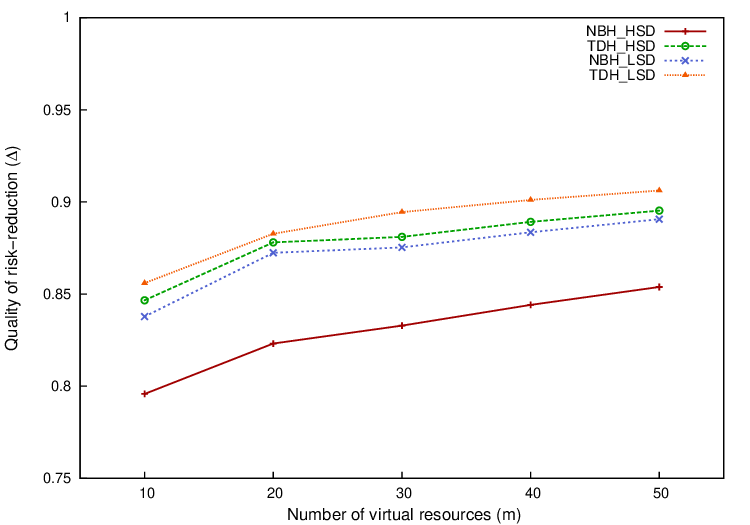}
        \caption{Quality of Risk-Reduction  $\Delta$ with a problem size of 150 roles  for LSD and HSD datacenters.} 
        \label{fig:150rQoRDiv}
        \vspace{-1.0em}
     \end{center}         
\end{figure} 

\subsubsection{Discrimination Index ($DI$)}
The $DI$ metric captures the behavior of both  heuristics  in terms of  $RISK(r_i)$. Intuitively, a low value of $DI$ indicates that the $RISK(r_i)$ is {\color{red}proportional} to $PA_i$.  On other hand, a high value of $DI$ indicates that the resulting $RISK(r_i)$ is {\color{red}disproportional} to $PA_i$ and, in essence, provides  an indication about  bad scheduling decisions.   Figure \ref{fig:30vDIDiv} depicts the performance of both heuristics for the cases of LSD and HSD as we increase the number of roles. Both heuristics tend to perform well in terms of $DI$. The reason being that increasing the number of roles tend to disperse the global sensitive property over large number of rules. As a consequence, the $PA_i$s tend to be close to each other, thus reducing the discrimination effect. In particular, we  note  that $DI$  for both heuristics improves (decreases) and eventually stays steady with no further improvement. A noticeable observation is that for $HSD$,  $DI$ always yields better results than $LSD$ since for small number of roles in $HSD$, $PA_i$'s have high values which making it  easier to generate  good scheduling decision in the initial assignment for the case of NBH as it uses a greedy strategy.  In fact,  for the case of HSD such "good decisions" are consistently made by $NBH$ with increasing value of $n$ as noticed by the green curve for 30 roles in Figure \ref{fig:30vDIDiv}. Similar reasoning can be provided for the behavior of TDH algorithm as shown in Figure \ref{fig:30vDIDiv}.

As we increase the number of VMs for both $LSD$ and $HSD$ datacenters,   $NBH$ tends to improve its $DI$ performance as depicted in Figure \ref{fig:150rDIDiv}., in contrast to $TDH$ where $DI$ is not  effected by the number of VMs. The reason is that  $NBH$ schedules roles using  a greedy approach  based on the $PA_i$ related to the roles. Therefore, increasing the number of VMs alway improves  $DI$ based performance of $NBH$.  In contrast, $TDH$  makes its  scheduling decisions using a more extensive search to generate a solution and  does not consider $PA_i$ during its assignment step.

\subsection{Performance Results for Mutual Information $MI$}

The performance of  of $NBH$ and $TDH$ for the sensitive property  $MI$ in terms of  $RISK_{MI}$,  $\Delta$, and $DI$ parameters is similar to the performance observed of the  $Div$ property. For example \textit{TDH} consistently outperforms $NBH$ as the number of roles and the number of VMs increase as shown in  Figures~\ref{fig:20vRiskMI},~\ref{fig:70rRiskMI}, and~\ref{fig:20vQoRMI}.  However,  we observe a higher value of $\Delta$ for $Div$ property (see Figure \ref{fig:20vQoRDiv}) as compared to value of $\Delta$ yields by $MI$ property. The main reason for  this difference is  the truncated  (up to level 3) sensitive property profile of RBAC policy  used for the evaluation of the proposed heuristics.

\begin{figure}[t!]
    \begin{center}
        \includegraphics[width=0.4\textwidth]{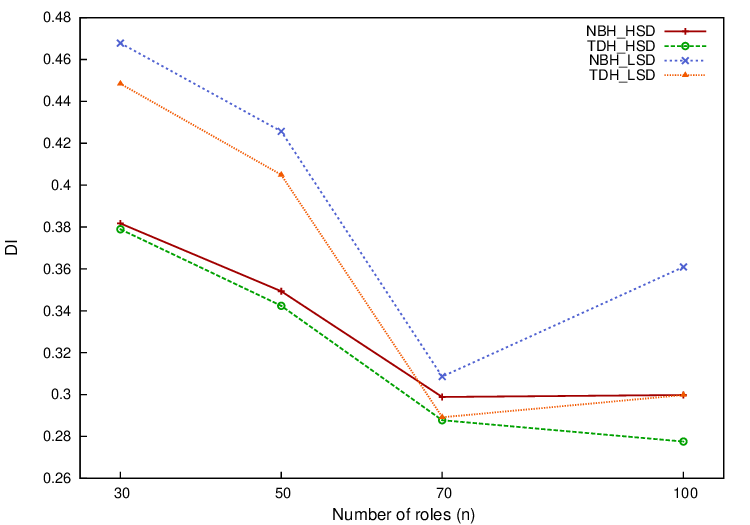}
        \caption{Discrimination Index  $DI$ with a problem size of 30 VMs  for LSD and  HSD datacenters.} 
        \label{fig:30vDIDiv}
        \vspace{-1.0em}
     \end{center}         
\end{figure} 

\begin{figure}[t!]
    \begin{center}
        \includegraphics[width=0.4\textwidth]{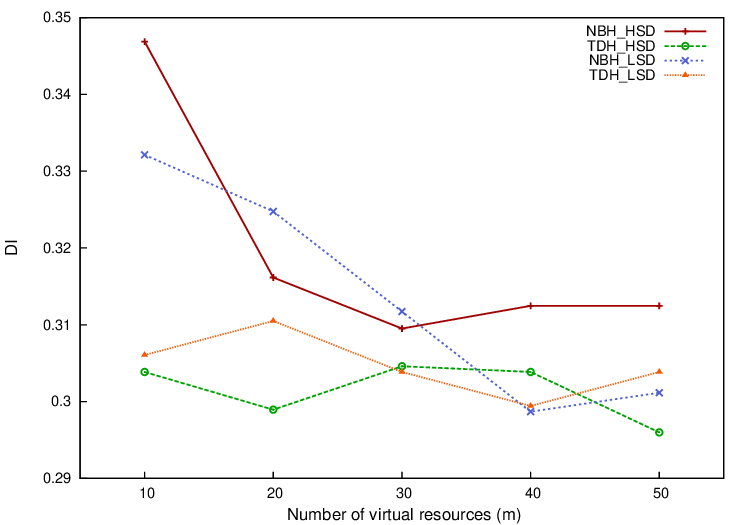}
        \caption{Discrimination Index  $DI$ with a problem size of 150 roles  for LSD and  HSD datacenters.} 
        \label{fig:150rDIDiv}
        \vspace{-2.0em}
     \end{center}         
\end{figure} 

\begin{figure}[t!]
    \begin{center}
        \includegraphics[width=0.4\textwidth]{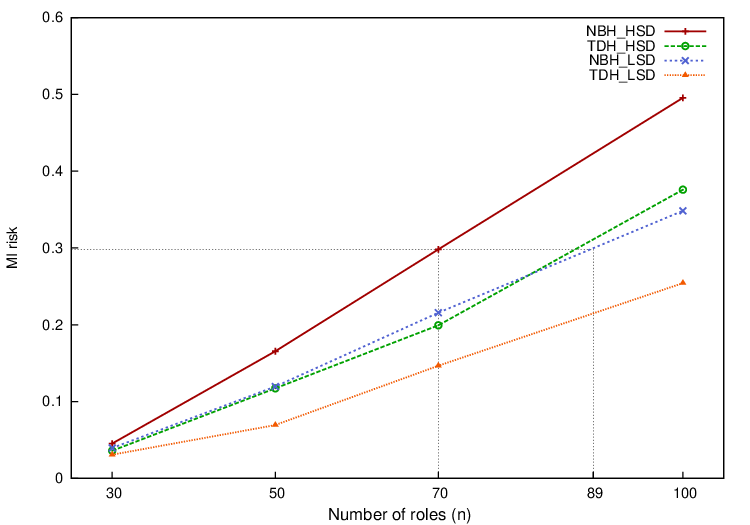}

        \caption{Mutual Information Risk $(RISK_{MI})$ with a problem size of 30 VMs  for LSD and HSD datacenters.} 
        \label{fig:20vRiskMI}
        \vspace{-1.0em}
     \end{center}         
\end{figure} 

\begin{figure}[t!]
    \begin{center}
        \includegraphics[width=0.4\textwidth]{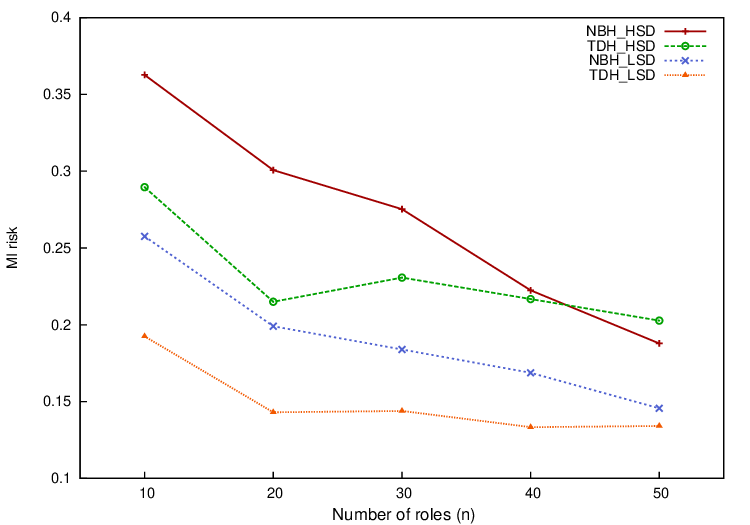}

        \caption{Mutual Information Risk $(RISK_{MI})$ with a problem size of 150 roles for LSD and HSD datacenters.} 
        \label{fig:70rRiskMI}
        \vspace{-2.0em}
     \end{center}         
\end{figure} 

\begin{figure}[t!]
    \begin{center}
        \includegraphics[width=0.4\textwidth]{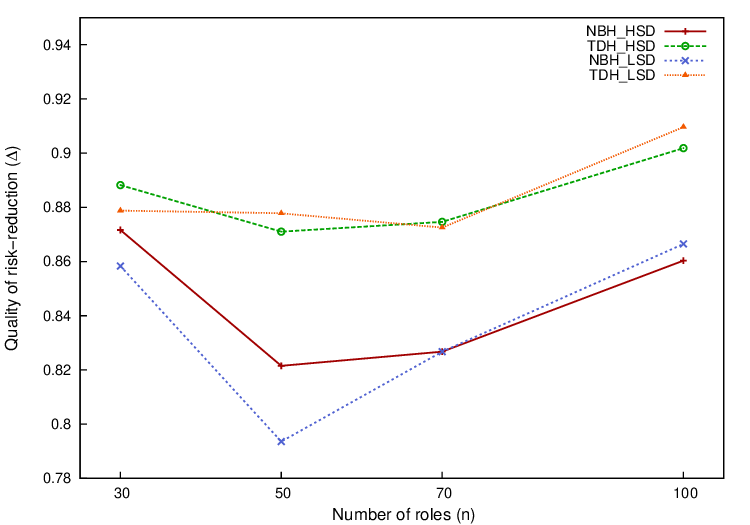}

        \caption{Quality of Risk-Reduction  $\Delta$ with a problem size of 30 VMs  for LSD, and HSD datacenters.} 
        \label{fig:20vQoRMI}
        \vspace{-2.0em}
     \end{center}         
\end{figure} 

 \begin{figure}[t!]
    \begin{center}
        \includegraphics[width=0.4\textwidth]{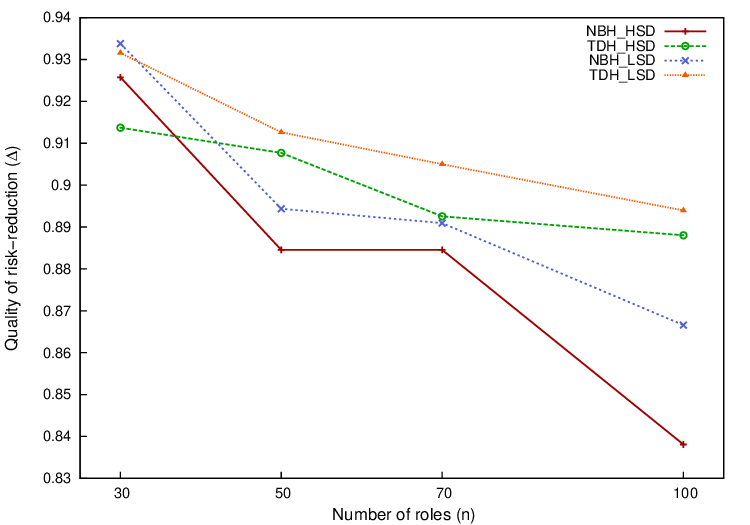}

        \caption{Quality of Risk-Reduction  $\Delta$ with a problem size of 30 VMs  for LSD, and HSD datacenters.} 
        \label{fig:20vQoRDiv}
        \vspace{-2.0em}
     \end{center}         
\end{figure}  

\begin{figure}[h!]
    \begin{center}
        \includegraphics[width=0.4\textwidth]{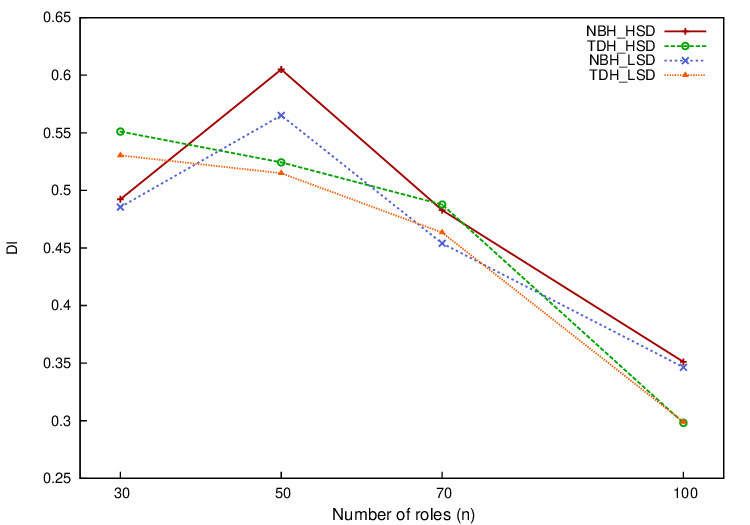}

        \caption{Discrimination Index  $DI$ with a problem size of 30 VMs  for LSD and  HSD datacenters.} 
        \label{fig:20vDIMI}
        \vspace{-2.0em}
     \end{center}         
\end{figure} 

In Figure \ref{fig:20vDIMI}, we note that irrespective of the heuristic, the value of $DI$ for the case of \textit{HSD} is higher when $n=50$ as compared to $n=30$ and $n=70$. To elaborate on this behavior, we show the $PA_i$  histograms for all roles in these three  cases of RBAC policy ($n=30$, $n=50$, and $n=70$)  for HSD as well as for LSD. Figures \ref{fig:LSDAttackMI},  and \ref{fig:HSDAttackMI}  provide these histograms. We sort the roles based on their $PA_i$ from large to small. The behavior of $DI$ curve in Figure 6.17 can be explained within the context of $PA_i$'s as following. The histogram for HSD for the case of n=50 shows a sharp drop of $PA_i$'s and a "piece-wise uniform" clusters of $PA_i$'s.  If two roles belong to the same  uniform cluster  are assigned to different type of VMs (in terms of vulnerability), it is expected  to increase the $DI$ metric.  As noted from  Figure 6.19 the uniform clustering phenomenon is more pronounced for the case of $n=50$; it is less pronounced for the case of $n=30$; and it is negligible for the case $n=70$. This results in the convex shaped behavior of the $DI$ graph of the HSD case. Similarly, for the case of LSD, as can be noticed from Figure 6.18, the uniform cluster behavior $PA_i$   is more noticeable  for smaller values of $n$ ( n=30 and n=50) and is negligible for $n=70$.

\begin{figure}[t!]
    \begin{center}
        \subfloat[n=30, LSD datacenter]{\label{fig:30rLSDPA}
        \includegraphics[width=0.4\textwidth]{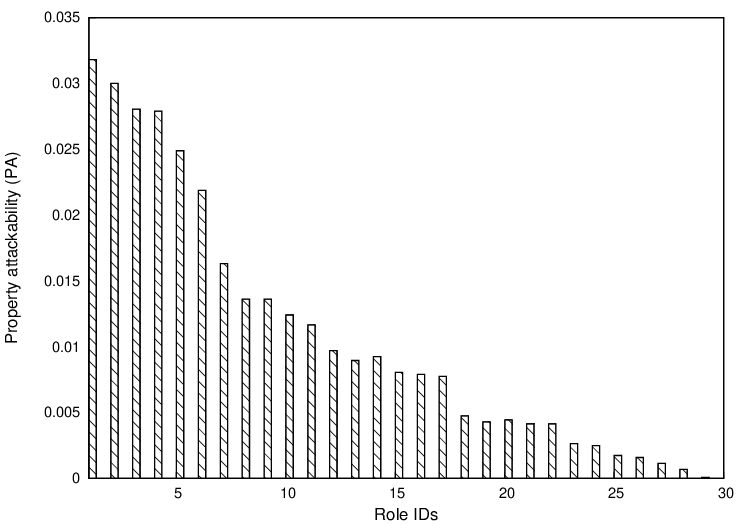}
        }
        \hskip 0.05truein

     \subfloat[n=50, LSD datacenter]{\label{fig:50rLSDPA}
        \includegraphics[width=0.4\textwidth]{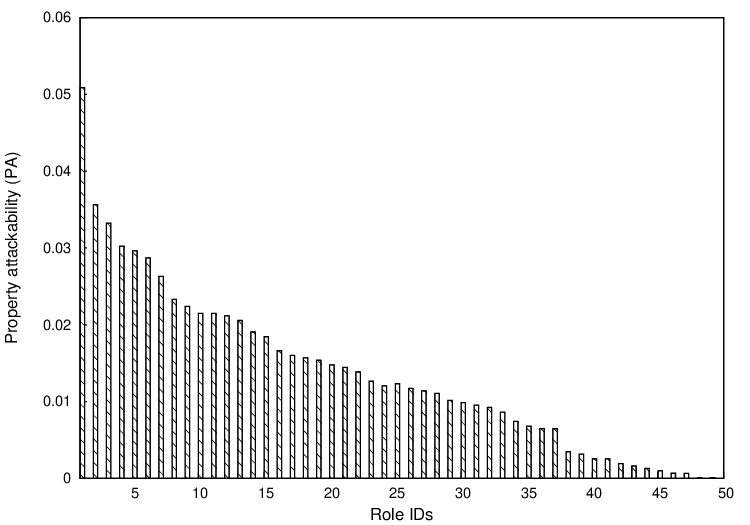}

        }
        \hskip 0.05truein
            \subfloat[n=70, LSD datacenter]{\label{fig:70rLSDPA}
        \includegraphics[width=0.4\textwidth]{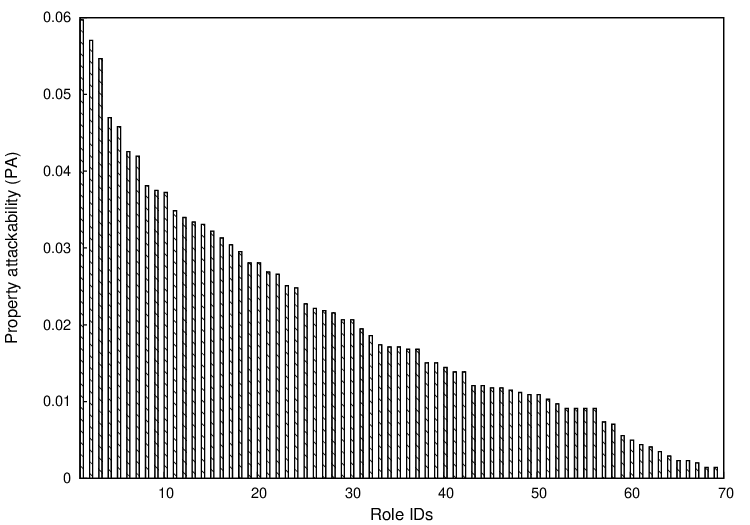}
        }
        \vspace{-1.0em}
        \caption{Attackability per role vs. $DI$ for LSD datacenter with mutual information sensitive property.}
\label{fig:LSDAttackMI}        
        \vspace{-2.0em}
       \end{center}         
\end{figure} 

\begin{figure}[h!]
    \begin{center}
        \subfloat[n=30, HSD datacenter]{\label{fig:30rHSDPA}
        \includegraphics[width=0.4\textwidth]{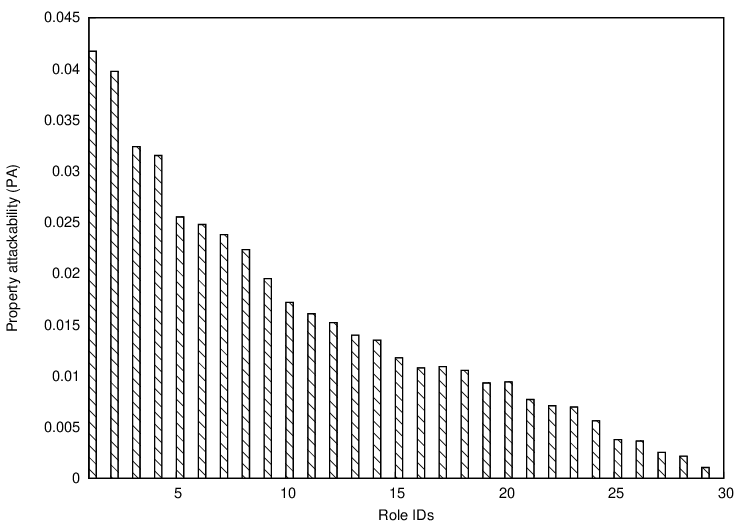}
        }
        \hskip 0.05truein

     \subfloat[n=50, HSD datacenter]{\label{fig:50rHSDPA}
        \includegraphics[width=0.4\textwidth]{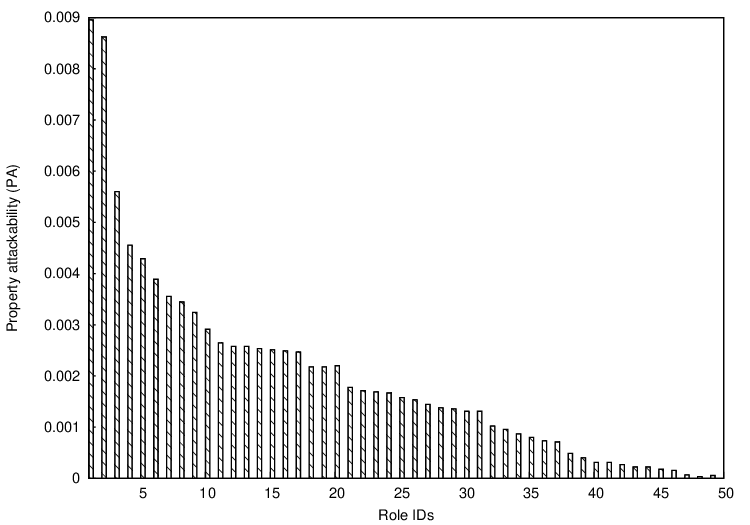}

        }
        \hskip 0.05truein
            \subfloat[n=70, HSD datacenter]{\label{fig:70rHSDPA}
        \includegraphics[width=0.4\textwidth]{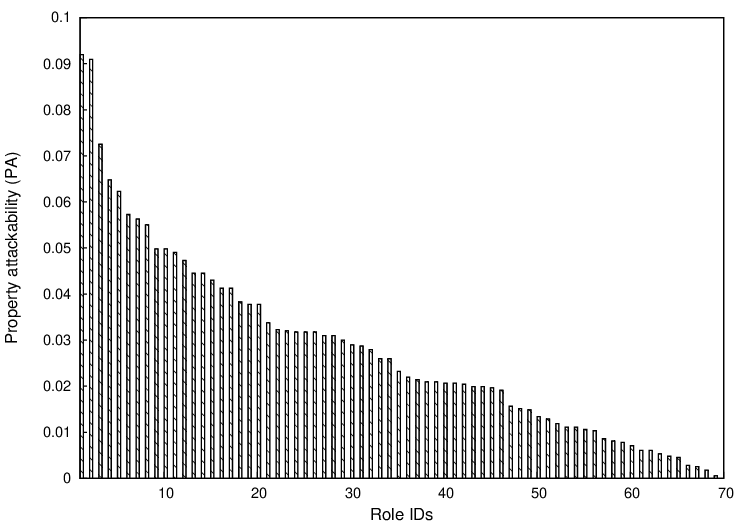}

        }

        \vspace{-1.0em}
        \caption{Attackability per role vs. $DI$ for HSD datacenter with mutual information sensitive property.}
\label{fig:HSDAttackMI}        
        \vspace{-2.0em}
       \end{center}         
\end{figure} 

\section{Related Work}\label{sec:related}
Work related to this paper can be categorized into three categories: 1) access control in cloud datacenter, 2) vulnerability models in cloud datacenter, and 3) security-aware resource assignment in cloud datacenter.

First, several researchers have addressed access control issues for cloud computing. Nurmi et al. \cite{nurmi2009eucalyptus} provided an authorization system to control the execution of VMs to ensure that only administrators and owners could access them. Berger et al. \cite{berger2009security} promoted an authorization model based on both RBAC and security labels to control access to shared data, VMs, and network resources. Calero et al. \cite{alcaraz10} presented a centralized authorization system that provides a federated path-based access control mechanism. What distinguishes our work is that we address the problems of virtual resource vulnerability in the presence of multitenancy and virtualization. 

Second, vulnerability discovery models (VDMs) are used to estimate the probability of software vulnerability. The authors in \cite{manadhata2011attack} proposed an attack surface metric that measures the security of software by analyzing its source code and find the potential flows. In \cite{alhazmi2007measuring} the vulnerability discovery is modeled based on the empirical analysis of history of actual vulnerability dataset and followed by predicting the vulnerability discovery. A Static-analysis of vulnerability indicator tool to assess the risk of software built by external deleopers is proposed in \cite{walden2012savi}, while \cite{shar2015web} proposes a hybrid (static and dynamic) analysis and machine learning techniques to extract vulnerabilities in web applications.  In \cite{massacci2014empirical}, an empirical methodology to systematically evaluate the performance of VDMs is proposed. The source code quality and complexity are incorporated in the VDM to generate better estimation \cite{rahimi2013vulnerability}. Common Vulnerability Scoring System (CVSS) \cite{mell2007complete} metrics have been extensively used by both industry and academia to score the vulnerabilities of the systems and estimate risks. The  US National Vulnerability Database (NVD) provides a catalog of known vulnerabilities with their CVSS scores. 

Finally, various security-aware scheduling techniques for cloud computing have been reported. In \cite{wen2016cost}, an algorithm to deploy workflow application over federated cloud is proposed. The proposed algorithm guarantees reliability and security constraints while optimizing the monetary cost. In particular, a workflow application is distributed on two or more cloud platforms in order to benefit from the advantage of each cloud. In order to prevent from unauthorized information disclosure of the deployed workflow over the federated cloud, a multi-level security model is used. In \cite{chen2017scheduling}, a task-scheduling technique for security sensitive workflows is proposed. The proposed technique addresses the problem of scheduling tasks with data dependencies, i.e., sensitive intermediate data. Authors presented theoretical guidelines for duplicating workload tasks in order to minimize the makespans and monetary costs of executing the workflows in the cloud. In \cite{chase2017scalable}, authors presented an approach for optimal provisioning of resources in the cloud using stochastic optimization. The provisioning technique takes into account future pricing, incoming traffic, and cyber attacks. The above work do not consider the data confidentiality of application or the access control policy, which restricts permissions for a given task. A key feature that differentiates our scheduling algorithms from other approaches is that the privilege set for each task is the main parameter in the scheduling decision. The objective of our scheduling algorithm is to minimize the risk of data leakage due to vulnerability of virtual resources in cloud computing environments.


%
%
%
%

\section{Conclusion}\label{sec:conclusion}
In this paper, we have proposed a sensitive property profile model for an RBAC policy. For this propose we have used two properties based on KL-divergence and mutual information extracted from check-in dataset. Based on the vulnerabilities of cloud architecture and the sensitive property profile, we have proposed resource scheduling  problem based on the  optimization pertaining to risk management. The problem is shown to be NP-complete. Accordingly, we have proposed two heuristics and presented their  simulation based performance results for $HSD$ and $LSD$ datacenters.
\bibliographystyle{IEEEtran}
\bibliography{references}
\end{document}